\documentclass[a4paper,10pt]{article}

\usepackage{graphicx}
\usepackage{longtable}
\usepackage{multirow}
\usepackage{float}
\usepackage{wrapfig}
\usepackage{rotating}
\usepackage[normalem]{ulem}
\usepackage{amsmath}
\usepackage{amsthm}
\usepackage{textcomp}
\usepackage{marvosym}
\usepackage{wasysym}
\usepackage{amssymb}
\usepackage{capt-of}
\usepackage[pdftex,                %
    bookmarks         = true,
    bookmarksnumbered = true,
    pdfstartview      = FitH,
    pdfpagelayout     = SinglePage,
    colorlinks        = true,
    linkcolor= blue, 
    anchorcolor= blue, 
    citecolor         =blue,
    urlcolor          = magenta,
    ]{hyperref}

\usepackage{comment}
\usepackage{color}
\usepackage{enumerate} 
\usepackage{myalgo}

\newcommand{\coq}{\text{\sc Coq}}
\newcommand{\hol}{\text{\sc Hol-light}}

\newcommand{\markus}[1]{{\color{blue}#1}}
\newcommand{\R}{\mathbb{R}}
\newcommand{\Z}{\mathbb{Z}}
\newcommand{\C}{\mathbb{C}}
\newcommand{\Q}{\mathbb{Q}}
\newcommand{\N}{\mathbb{N}}

\DeclareMathOperator{\bigo}{\mathcal{O}}
\DeclareMathOperator{\sgn}{sgn}
\DeclareMathOperator{\bigotilde}{\overset{\sim}{\mathcal{O}}}
\newcommand{\hasrealroots}{\texttt{has\_real\_roots}}
\newcommand{\hasrealrootsfun}[1]{\texttt{has\_real\_roots}(#1)}
\newcommand{\sumtwosquares}{\texttt{sum\_two\_squares}}
\newcommand{\sumtwosquaresfun}[2]{\texttt{sum\_two\_squares}(#1,#2)}

\newcommand{\qlist}{\texttt{q\_list}}
\newcommand{\hlist}{\texttt{h\_list}}
\newcommand{\clist}{\texttt{c\_list}}
\newcommand{\slist}{\texttt{s\_list}}
\newcommand{\univsos}{\texttt{univsos}}
\newcommand{\univsosone}{\texttt{univsos1}}
\newcommand{\univsostwo}{\texttt{univsos2}}
\theoremstyle{plain}
\newtheorem{theorem}{Theorem}[section]
\newtheorem{lemma}[theorem]{Lemma}
\newtheorem{proposition}[theorem]{Proposition}

\theoremstyle{definition}

\newtheorem{remark}[theorem]{Remark}

\newtheorem{example}{Example}

\if{
\newcommand{\qed}{\nobreak \ifvmode \relax \else
      \ifdim\lastskip<1.5em \hskip-\lastskip
      \hskip1.5em plus0em minus0.5em \fi \nobreak
      \vrule height0.75em width0.5em depth0.25em\fi}
}\fi

\def\mohab#1{\textcolor{magenta}{#1}}
\definecolor{dkviolet}{rgb}{0.6,0,0.8}

\begin{document}
\author{Victor Magron$^{1}$ \and Mohab Safey El Din$^{2}$ \and Markus Schweighofer$^{3}$}
\date{\today}
\title{Algorithms for Weighted Sums of Squares Decomposition of
Non-negative Univariate Polynomials}

\footnotetext[1]{CNRS Verimag; 700 av Centrale 38401 Saint-Martin d'H\`eres, France} 
\footnotetext[2]{Sorbonne Universit\'es, UPMC Univ. Paris 06, CNRS, Inria Paris Center, LIP6, Equipe PolSys, F-75005, Paris, France}
\footnotetext[3]{{Fachbereich Mathematik und Statistik, Universit\"{a}t Konstanz, 78457 Konstanz, Germany}}
\maketitle



\begin{abstract}
  It is well-known that every non-negative univariate real
  polynomial can be written as the sum of two polynomial squares with
  real coefficients. When one allows a \emph{weighted} sum of
  \emph{finitely many} squares instead of a sum of two squares, then
  one can choose all coefficients in the representation to lie in the
  field generated by the coefficients of the polynomial. In
    particular, this allows an effective treatment of polynomials with
    rational coefficients.

  In this article, we describe, analyze and compare
  both from the theoretical and practical points of view, two
  algorithms computing such a weighted sums of squares decomposition
  for univariate polynomials with rational coefficients.

  The first algorithm, due to the third author  relies on real root isolation,
  quadratic approximations of positive polynomials and square-free
  decomposition but its complexity was not analyzed. 
  We provide bit complexity estimates, both on the runtime and the
  output size of this algorithm. They are exponential in the degree
  of the input univariate polynomial and linear in the maximum bitsize
  of its complexity. This analysis is obtained using quantifier
  elimination and root isolation bounds.

The second algorithm, due to Chevillard, Harrison, Joldes and Lauter,
 relies on complex root
isolation and square-free decomposition and has been introduced for
certifying positiveness of polynomials in the context of computer
arithmetics. Again, its complexity was not analyzed. 
We provide bit complexity estimates, both on the runtime and the
output size of this algorithm, which are polynomial in the degree of
the input polynomial and linear in the maximum bitsize of its
complexity. This analysis is obtained using Vieta's formula and root
isolation bounds.

Finally, we report on our implementations of both algorithms and
compare them in practice on several application benchmarks. While the
second algorithm is, as expected from the complexity result, more
efficient on most of examples, we exhibit families of non-negative
polynomials for which the first algorithm is better. 
\end{abstract}

\section{Introduction}
\label{sec:intro}
Given a {subfield} $K$ of $\R$ and a non-negative univariate polynomial
$f \in K[X]$, we consider the problem of proving the existence and
computing the weighted sums of squares decompositions of $f$ with
coefficients also lying in $K$.

Beyond the theoretical interest of this question, finding certificates
of non-negative polynomials is mandatory in many application
fields. Among them, one can mention the stability proofs of critical
control systems often relying on Lyapunov functions~\cite{Rantzer00},
the certified evaluation of mathematical functions in the context of
computer arithmetics (see for instance~\cite{Chevillard11}), the
formal verification of real inequalities~\cite{kepler15} within proof
assistants such as $\coq$~\cite{coq} or $\hol$~\cite{hollight} ; in
these situations the univariate case is already an important one.
{In particular, formal proofs of polynomial non-negativity can
  be handled with sums of squares certificates. These certificates are
  obtained with tools available outside of the proof assistants and
  eventually verified inside. Because of the limited computing power
  available inside such proof assistants, this is crucial to devise
  algorithms that produce certificates, whose checking is
  computationally reasonably simple. In particular, we would like to
  ensure that such algorithms output sums of squares certificates of
  moderate bitsize and ultimately with a computational complexity
  being polynomial with respect to the input.  }


\paragraph*{Related Works.}
Decompositing non-negative univariate polynomials into sums of squares
has a long story ; very early quantitative aspects like the number of
needed squares have been studied. For the case $K = \Q$, Landau shows
in~\cite{Landau1906} that for every non-negative polynomial in $\Q[X]$,
there exists a decomposition involving a weighted sum of (at most)
eight polynomial squares in $\Q[X]$. In~\cite{Pourchet71}, Pourchet
improves
this result by showing the existence of a decomposition involving only
a weighted sum of (at most) five squares. This is done using
approximation and valuation theory ; extracting an algorithm from
these tools is not the subject of study of this
paper. 
%
\if{ To the best of our knowledge, the first answer to
  Question~\ref{th:sosK} is the Nichtnegativstellens\"atz result for
  univariate polynomials given by Pourchet in~\cite{Pourchet71} for
  the case $K = \Q$, where it is shown that for every non-negative
  polynomial in $\Q[X]$, there exists a decomposition involving a sum
  of (at most) five polynomial squares in $\Q[X]$.  An English
  translation is provided in~\cite[Chapter 17]{Rajwade93}.
  \markus{Pourchet lui même écrit que Landau
    [\url{https://www.math.ucdavis.edu/~deloera/MISC/BIBLIOTECA/trunk/Landau.pdf}]
    a déjà répondu à cette question
    [\url{http://matwbn.icm.edu.pl/ksiazki/aa/aa19/aa1917.pdf}]}
  Computing sum of squares decompositions with rational coefficients
  H. Peyrl and P.A. Parrilo Theoretical Computer Science, Vol. 409,
  Issue 2, pp. 269-281 (2008)
  [http://www.sciencedirect.com/science/article/pii/S0304397508006452?via
}\fi 

\if{
\mohab{\sout{In the multivariate case, alternative methods in Computer
    Algebra have been developed for algebraic certification of
    polynomial non-negativity as well as for polynomial
    optimization. Among them, we can mention approaches relying on
    critical point methods~\cite{Greuet11}, Cylindrical Algebraic
    Decomposition~\cite{Iwane13} and algebraic solutions of linear
    matrix inequalities (LMIs)~\cite{Kaltofen12} and \cite{Guo13}. The
    result from}}~\cite[Theorem 1.1]{Guo13} \mohab{\sout{states the
    existence of an algorithm to compute a sums of squares (SOS)
    decomposition of $f$ for the case $K = \Q$, with output bitsize
    bounded by}} $\tau^{\bigo{(1)}} 2^{\bigo{(n^3)}}$\mohab{\sout{,
    with $\tau$ being an upper bound on the coefficient bitsizes of
    $f$. Recent research efforts~\cite{Simone16} rely on
    determinantial varieties to solve LMIs but are currently
    restricted to problems of small size.}}

\mohab{\bf todo: parler des approches symboliques: celles qui
  conduisent \`a des solutions alg\'ebriques et celles qui conduisent
  \`a des solutions rationnelles.} 
}\fi  

More recently, the use of semidefinite programming for computing
sums-of-squares certificates of non-negativity for polynomials has
become very popular since \cite{Lasserre01,Parrilo00}.  Given a
polynomial $f$ of degree $n$, this method consists in finding a real
symmetric matrix $G$ with non-negative eigenvalues (a {\em semidefinite
  positive} matrix), such that $f(x) = v(x)^T G v(x)$, where $v$ is
the vectors of monomials of degree less than $n/2$. Hence, this leads
to solve a so-called Linear Matrix Inequality and one can rely on
semidefinite programming (SDP) to find the coefficients of $G$. This
task can be delegated to an SDP solver (e.g. SeDuMi, SDPA, SDPT3 among
others). An important technical issue arises from the fact that such
SDP solvers are most of the time implemented with floating-point
double precision. More accurate solvers are available
(e.g. SDPA-GMP~\cite{Nakata10GMP}). However, note that these solvers
always compute numerical approximations of the algebraic solution to
the semidefinite program under consideration. Hence, they are not
sufficient to provide algebraic certificates of posivity with rational
coefficients. Hence, a process for making exact and with rational
coefficients the computed numerical approximations of sums of squares
certificates is needed.  This issue has been tackled
in~\cite{Peyrl08,Kaltofen12}. The certification scheme described
in~\cite{Magron14} allows to obtain lower bounds of non-negative
polynomials over compact sets.  However, despite their efficiency,
these methods do not provide any guarantee to output a rational
solution to a Linear Matrix Inequality when it exists (and especially
when it is far from the computed numerical solution).

A more systematic treatment of this problem has been brought by the
symbolic computation community. Linear Matrix Inequalities can be
solved as a decision problem over the reals with polynomial
constraints using the Cylindrical Algebraic Decomposition algorithm
\cite{Collins75} or more efficient critical point methods (see
e.g. \cite{BPR96} for complexity estimates and \cite{HS12, GS14} for
practical algorithms). But using such general algorithms is an
overkill and dedicated algorithms have been designed for computing
exactly algebraic solutions to Linear Matrix Inequalities
\cite{Simone16,SPECTRA}. Computing rational solutions can also be
considered thanks to convexity properties \cite{SZ10}. In particular,
the algorithm in \cite{Guo13} can be used to compute sums of squares
certificates with rational coefficients fpr a non-negative univariate
polynomial of degree $n$ with coefficients of bit size bounded by
$\tau$ using $\tau^{\bigo{(1)}} 2^{\bigo{(n^3)}}$ bit operations at
most (see \cite[Theorem 1]{Guo13}). 
%

{For the case where $K$ is an arbitrary subfield of $\R$},
Schweighofer gives in~\cite{SchweighoferMasterThesis} a new proof of
the existence of a decomposition involving a sum of (at most) $n$
polynomial squares in $K[X]$. This existence proof comes together with
a recursive algorithm to compute such decompositions.  At each
recursive step, the algorithm performs real root isolation and
quadratic approximations of positive polynomials.  Later on, a second
algorithm is derived in~\cite[Section~5.2]{Chevillard11}, where the
authors show the existence of a decomposition involving a sum of (at
most) $n+3$ polynomial squares in $K[X]$. This algorithm is based on
approximating complex roots of perturbed positive polynomials.

These both latter algorithms were not analyzed despite the fact that
they were implemented and used. An outcome of this paper is a bit
complexity analysis for both of them, showing that they have better
complexities than the algorithm in \cite{Guo13}, the second algorithm
being polynomial in $n$ and $\tau$.

  \paragraph*{Notation for complexity estimates.}  
  {For complexity estimates, we use the bit complexity model.}
  For an integer $b \in \Z \backslash \{0\}$, we denote by
  $\tau(b) := \log_2 (|b|)+1$ the bitsize of $b$, with the convention
  $\tau(0) := 1$.  We write a given polynomial $f \in \Z[X]$ of degree
  $n \in \N$ as $f = \sum_{i=0}^n b_i X^i$, with
  $b_0, \dots, b_n \in \Z$.  In this case, we define
  $\| f \|_{\infty} := \max_{0 \leq i \leq n} |b_i|$ and, using a
  slight abuse of notation, we denote $\tau (\| f \|_{\infty})$ by
  $\tau(f)$. Observe that when $f$ has degree $n$, the bit size
  necessary to encode $f$ is bounded by $n\tau(f)$ (when storing each
  coefficients of $f$). {The derivative of $f$ is
    $f' = \sum_{i=1}^{n} i b_i X^{i-1}$.}
  For a rational number $q = \frac{b}{c}$, with
  $b \in \Z, c \in \Z \backslash \{0\}$ and ${\rm gcd}(b,c)=1$, we
  denote $\max \{\tau (b), \tau (c)\}$ by $\tau (q)$.
For two mappings $g, h : \N^l \to \R$, the expression
``$g(v) = \bigo{(h(v))}$'' means that there exists a integer
$b \in \N$ such that for all $v \in \N^l$, $g(v) \leq b h(v)$. The
expression ``$g(v) = \bigotilde{(h(v))}$'' means that there exists a
integer $c \in \N$ such that for all $v \in \N^l$,
$g(v) \leq h(v) \log_2 (h(v))^c$.

\paragraph*{Contributions.}
We present and analyze two algorithms, denoted by~$\univsosone$ and
$\univsostwo$, allowing to decompose a non-negative univariate
polynomial $f$ of degree $n$ into sums of squares with coefficients
lying in any subfield $K$ of $\R$.  {To the best of our knowledge, there was no prior complexity estimate for the output of such certification algorithms based on sums of squares in the univariate case.}
We summarize our contributions as follows:
\begin{itemize}
\item We  describe in Section~\ref{sec:method1}
  the first algorithm, called~$\univsosone$ in the sequel. It was
  originally given in~\cite[Chapter 2]{SchweighoferMasterThesis} ;
  Section~\ref{sec:method1} can be seen as a partial English translation of
    this text written in German since some proofs have been
    significantly simplified.
 In the same section, we analyze its bit complexity. When the input
    is a polynomial of degree $n$ with integer coefficients of maximum bitsize $\tau$,
   %
    we prove that Algorithm~$\univsosone$ uses
    $\bigotilde{((\frac{n}{2})^{\frac{3n}{2}} \tau)}$ boolean
    operations and returns polynomials of bitsize bounded by
    $\bigo{((\frac{n}{2})^{\frac{3n}{2}} \tau)}$. This is not
    restrictive: when $f\in \Q[X]$, one can multiply it by the least
    common multiple of the denominators of its coefficients and apply
    our statement for polynomials in $\Z[X]$.
  
%
  \item We describe in Section~\ref{sec:method2} the second
    algorithm~$\univsostwo$, initially given
    in~\cite[Section~5.2]{Chevillard11}. We also analyze its {bit}
    complexity.
%
{When the input is
  a univariate polynomial of degree $n$ with integer coefficients of
  maximum bitsize $\tau$, we prove that Algorithm $\univsostwo$
  returns a sums of square decompositions of $n+3$ polynomials with
  coefficients of bitsize bounded by $\bigo{(n^3 + n^2 \tau)}$ using
  $ \bigotilde{(n^4 + n^3 \tau )}$ boolean operations.}
\item Both algorithms are implemented within the $\univsos$ tool. The
  first release of $\univsos$ is a Maple library, is freely
  available\footnote{\url{https://github.com/magronv/univsos}} {
    and is integrated in the RAGlib (Real Algebraic Library) Maple
    package\footnote{\url{http://www-polsys.lip6.fr/~safey/RAGLib/}}}. The
  scalability of the library is evaluated in Section~\ref{sec:benchs}
  on several non-negative polynomials issued from the existing
  literature. Despite the significant difference of theoretical
  complexity between the two algorithms, numerical benchmarks indicate
  that both may yield competitive performance w.r.t.~specific
  sub-classes of problems.
\end{itemize}

\section{Preliminaries}
\label{sec:back}
We first recall the proof of the following classical result for
non-negative real-valued univariate polynomials (see e.g.~\cite[Section
8.1]{prestel2001positive}).

\begin{theorem}
\label{th:sosR}
Let $f \in \R[X]$ be a non-negative univariate polynomial,
i.e.~$f(x) \geq 0$, for all $x \in \R$. Then, $f$ can be written as
the sum of two polynomial squares in $\R[X]$.
\end{theorem}
\begin{proof}
  Without loss of generality, one can assume that $f$ is monic,
  i.e.~the leading coefficient (nonzero coefficient of highest degree)
  of $f$ is $1$. Then we decompose $f$ as follows in $\C[x]$:
  \[ 
  f = \prod_j (X - a_j)^{r_j} \prod_k ((X - (b_k + i c_k)) (X - (b_k - i c_k)))^{s_k} \,,
  \]
  with $a_j, b_k, c_k \in \R$, $r_j, s_k \in \N^{>0}$, $a_j$ standing
  for the real roots of $f$ and $(b_k \pm i c_k)$ standing for the
  complex conjugate roots of $f$.  Since $f$ is non-negative, all real
  roots must have even multiplicity $r_j$, yielding the existence of
  polynomials $g, q, r \in \R[X]$ satisfying the following:
  \[
  g^2 = \prod_j (X - a_j)^{r_j} \,, \quad q + i r = \prod_k (X -
  (b_k + i c_k))^{s_k} \, q - i r = \prod_k (X -
  (b_k - i c_k))^{s_k}.
  \]
  Then, one has
  $f = g^2 (q + i r) (q - i r) = g^2 (q^2 + r^2) = (g q)^2 + (g r)^2$,
  which proves the claim.
\end{proof}
%
  
  Let $K$ be a field and $g \in K[X]$. One says that $g$ is a
  {\em square-free} polynomial, when there is no prime element
  $p \in K[X]$ such that $p^2$ divides $g$.
Let now $f\in K[X]-\{0\}$.  A decomposition of $f$ of the form
  $f = a g_1^1 g_2^2 \dots g_n^n$ with $a \in K$ and normalized
  pairwise coprime square-free polynomials $g_1, g_2,\dots, g_n$ is
  called a {\em square-free decomposition} of $f$ in $K[X]$. 
\if{
%
%
The existence of square-free decompositions in $K[X]$ is ensured from
the existence of prime factor decompositions in $K[X]$. 

\mohab{M\^eme remarque: c'est trop standard et pas assez utilis\'e dans la suite.}
\begin{lemma}
\label{th:sfexistence}
Let $K$ be a field and $f \in K[X]\mohab{-\{0\}}$.
\begin{enumerate}[(i)]
\item There exists a square-free decomposition of $f$ in $K[X]$.
\item With $f = a g_1^1 g_2^2 \dots g_n^n$ being such a decomposition,
  for all $i \in \{1, \dots, n \}$, $g_i$ is the product of all
  normalized prime factors $p$ of $f$ in $K[X]$ satisfying
  $\max \{ \alpha \in \N \mid p^\alpha \text{ divides } f \} = 1$.
\end{enumerate}
\end{lemma}
For fields of characterisic 0, existence of square-free decompositions
comes the existence of prime factor decompositions.
\mohab{\sout{(see~\cite{Grobner93} for fields of nonzero
    characteristics).}}

\mohab{M\^eme remarque: on ne se sert pas de ce lemme dans la suite.}
\begin{lemma}
\label{th:primechar0}
Let $K$ be a field of characteristic 0. Let $f \in K[X]\mohab{-\{0\}}$
and $p$ a prime factor of $f$. Then, one has:
\[
\max \{ \alpha \in \N \mid p^\alpha \text{ divides } f \} - 1
= 
\max \{ \alpha \in \N \mid p^\alpha \text{ divides } f' \} \,.
\]
\end{lemma}
\begin{proof}
Write $f = p^\alpha g$ for some $g \in K[X]$ such that $p$ does not divide $g$. In addition, one has:
\[
f' = p^\alpha g + \alpha p^{\alpha - 1} p' g = p^{\alpha - 1} (p g' + \alpha p' g) \,.
\]
To prove our claim, it is enough to show that $p$ does not divide
$p g' + \alpha p' g$. Assume that $p$ divides $p g' + \alpha p'
g$.
Then, $p$ also divides $\alpha p' g$. Since $K$ has characteristic 0,
then $\alpha \neq 0$ in $K$, implying that $p$ divides $p' g$. But $p$
is a prime element which does not divide $g$, hence $p$ must divide
$p'$. Finally, $p'$ has to be the zero polynomial, which contradicts
the fact that $K$ is a field of characteristic 0.
\end{proof}
\mohab{\sout{Let $K$ be a field of characteristic 0. We briefly explain how to
obtain the square-free decomposition $g_1^1 g_2^2 \dots g_n^n$, with
normalized pairwise coprime square-free polynomials
$g_1, g_2,\dots, g_n$. Thanks to the above lemma and
Lemma~\ref{th:sfexistence}~(ii), the greatest common divisor of
$g_1^1 g_2^2 \dots g_n^n$ and $(g_1^1 g_2^2 \dots g_n^n)'$ is
$g_2^1 \dots g_n^{n-1}$. Dividing $g_1^1 g_2^2 \dots g_n^n$ by
$g_2^1 \dots g_n^{n-1}$ yields the product $g_1 g_2 \dots
g_n$.
Applying the same procedure to $g_2^1 \dots g_n^{n-1}$ allows to
obtain the product $g_2 \dots g_n$. Finally, one obtains the product
list $g_1 \dots g_n$, $g_2 \dots g_n$, $g_{n-1} g_n$, $g_n$ and
successive divisions yield the polynomials $g_1, \dots, g_n$.}}
}\fi

{
We recall the following useful classical bounds.
\begin{lemma}{~\cite[Corollary~10.12]{BPR06}}
\label{th:bitdiv}
If $p \in \Z[X]$ and $q \in \Z[X]$ divides $p$ in $\Z[X]$, then one has $\tau(q) \leq \deg q + \tau(p) + \log_2(\deg p + 1)$.
\end{lemma}
\if{
\begin{lemma}{~\cite[Corollary 11.11]{Gathen99}}
\label{th:costgcd}
Let $f, g \in \Z[X]$ with degree at most $n$ with coefficient bitsize upper bounded by $\tau$. Then the gcd of $f$ and $g$ can be computed using an expected number of $ \bigotilde{(n^2 + n \tau )}$ boolean operations.
\end{lemma}
}\fi
The algorithm of Yun~\cite{Yun76} (also described in~\cite[Algorithm~14.21]{Gathen99}) allows to compute a square-free decomposition of polynomials with coefficients in a field of characteristic 0. 
%
\begin{lemma}{~\cite[\S~11.2]{Gathen99}}
\label{th:costsqrfree}
Let $f \in \Z[X]$ with degree at most $n$ with coefficient bitsize upper bounded by $\tau$. Then the square-free decomposition of $f$ using Yun's Algorithm~\cite{Yun76} can be computed using an expected number of $\bigotilde{(n^2 \tau )}$ boolean operations.
\end{lemma}
}
%
%
\if{
The computational cost of this procedure to obtain a square-free
decomposition is the same for any field extension of $K$.  
\mohab{\bf
  todo: donner le co\^ut de ce calcul car on en a besoin pour analyser
  les complexit\'es; pour le coup, quand ce sera fait on utilisera le
  lemme qui donne les complexit\'es (runtime + taille de la sortie).}
}\fi
{
\begin{lemma}{~\cite[\S~6.3.1]{Mignotte1992}}
\label{th:sfextend}
Let $K$ be a field of characteristic 0 and $L$ a field extension of
$K$. The square-free decomposition in $L[X]$ of any polynomial
$f \in K[X]-\{0\}$ is the same as the square-free
decomposition of $f$ in $K[X]$. Any polynomial $f \in K[X]-\{0\}$ which is a square-free
polynomial in $K[X]$ is also a square-free polynomial in $L[X]$.
\end{lemma}
}
  \if{
\begin{lemma}
\label{th:sfextend2}
Let $K$ be a field of characteristic 0 and $L$ a field extension of
$K$. Any polynomial $f \in K[X]\mohab{-\{0\}}$ which is a square-free
polynomial in $K[X]$ is also a square-free polynomial in $L[X]$.
\end{lemma}
}\fi

The following lemma allows to obtain upper bounds on the magnitudes of
the roots of a univariate polynomial.
%
 %
\begin{lemma}{\textbf{(Cauchy Bound~\cite{Cauchy1830})}}
\label{th:cauchy}
Let $K$ be an ordered field. Let $a_0, \dots, a_n \in K$ with
$a_n \neq 0$. Let $x \in K$ such that $\sum_{i=0}^n a_i x^i =
0$. Then, one has:
\[
|x| \leq \max \left\{ 1, \dfrac{|a_0|}{|a_n|} + \dots + \dfrac{|a_{n-1}|}{|a_n|} \right\} \,.
\]
\end{lemma}
For polynomial with integer coefficients, one has the following:
\begin{lemma}{~\cite{Mignotte1992}}
\label{th:cauchyint}
Let $f \in \Z[X]$ of degree $n$, with coefficient bitsize upper bounded by $\tau$. If $f(x) = 0$ and $x \neq 0$, then $\frac{1}{2^\tau + 1} \leq |x| \leq 2^\tau + 1$.
\end{lemma}
{ The real (resp.~complex) roots of a polynomial can be
  approximated using root isolation techniques. To compute the real
  roots one can use algorithms based on Uspensky's method relying on
  Descartes's rule of sign, see e.g.~\cite[Chap. 10]{BPR06} for a
  general description of real root isolation algorithms.
\begin{lemma}{~\cite[Theorem~5]{Mehlhorn15}}
\label{th:rootisolation}
Let $f \in \Z[X]$ with degree at most $n$ with coefficient bitsize
upper bounded by $\tau$. Isolating intervals (resp.~disks) of radius
less than $2^{-\kappa}$ for all real (resp.~complex) roots of $f$ can
be computed in $\bigotilde{(n^3 + n^2 \tau + n \kappa)}$ boolean
operations.
\end{lemma}
%
Vieta's formulas provide relations between the coefficients of a
polynomial and signed sums and products of the complex roots of this
polynomial:
\begin{lemma}{\textbf{(Vieta's formulas~\cite{girard1629})}}
\label{eq:vieta}
Let $K$ be an ordered field. Given a polynomial $f = \sum_{i=0^n} a_i X^i \in K[X]$ with $a_n \neq 0$ with (not necessarily distinct) complex roots $z_1, \dots, z_n$, one has for all $j=1,\dots,n$:
\[
\sum_{1 \leq i_1 < \dots < i_j \leq n} z_{i_1} \dots z_{i_j} = (-1)^j \frac{a_{n-j}}{a_n} \,.
\]
\end{lemma}
}
\if{
\begin{proof}
  Assume that $|x| \geq 1$. Then, one has
  $|a_n x^n| = | \sum_{i=0}^{n-1} a_i x^i | \leq \sum_{i=0}^{n-1}
  |a_i| \, |x|^i \leq \sum_{i=0}^{n-1} |a_i| \, |x|^{n-1} $.
  Dividing by $|a_n| \, |x|^{n-1}$ concludes the proof.
\end{proof}
}\fi
%


\section{Nichtnegativstellens\"atze with quadratic approximations}
\label{sec:method1}
\subsection{A proof of the existence of SOS decompositions}
\label{sec:proof1} 
\begin{lemma}
\label{th:parab}
Let $K$ be an ordered field. Let $g = a X^2 + b X + c \in K[X]$ with
$a,b,c \in K$ {and} $a \neq 0$. {Then, $g$ can be rewritten as}
$g = a\left (X + \frac{b}{2 a}\right )^2 + \left (c -
  \frac{b^2}{4a}\right )$.
{Moreover, }when $g$ is non-negative over $K$, one has $a > 0$ and
$c - \frac{b^2}{4a} \geq 0$.
\end{lemma}
\begin{proof}
  The decomposition of $g$ is straightforward. Assume that $g$ is
  non-negative over $K$. Remark that
  $c - \frac{b^2}{4a} = g\left (-\frac{b}{2a}\right )$ ; hence since
  we assume that $g$ is non-negative over $K$ we deduce that
  $c-\frac{b^2}{4a}\geq 0$. 

  It remains to prove that $a>0$ which we do by contradiction,
  assuming that $a < 0$. Then, this implies that for all $x \in K$,
  one has
  $\left (x + \frac{b}{2 a}\right )^2 \leq -\frac{1}{a}\left (c -
    \frac{b^2}{4a}\right )$.
  Thus, there exists $C \in K$ such that $x^2 \leq C$, for each
  $x \in K$.  This implies in particular for $x = 2$ that $4 \leq C$
  and for $x = C$ that $C^2 \leq C$, thus $C \leq 1$. Finally, one
  obtains $4 \leq C \leq 1$, yielding a contradiction.
\end{proof}
%
Let $f\in K[X]$ be a square-free polynomial which is non-negative over
$\R$. Then, $f$ is positive over $\R$, otherwise $f$ would have at
least one real root, implying that $f$ would be neither a square-free
polynomial in $\R[X]$ nor a square-free polynomial in $K[X]$,
according to Lemma~\ref{th:sfextend}.
We want to find a polynomial $g \in K[X]$ which fulfills the following
conditions:
\begin{enumerate}[(i)]
\item $\deg g \leq 2$,
\item $g$ is non-negative over $\R$,
\item $f - g$ is non-negative over $\R$,
\item $f - g$ has a root $t \in K$.
\end{enumerate}
Assume that Property (i) holds. Then the existence of a sum of squares
decomposition for $g$ is ensured from Property (ii). Property (iii)
implies that $h = f - g$ has only non-negative values over $\R$.  The
aim of Property (iv) is to ensure the existence of a root $t \in K$ of
$h$, which is stronger than the existence of a real root. Note that
the case where the degree of $h = f -g $ is less than the degree of
$f$ occurs only when $\deg f = 2$. In this latter case, we can
rely on Lemma~\ref{th:parab}.

Now, we investigate the properties of a polynomial $g \in K[X]$, which
fulfills conditions (i)-(iii) and (iv) with $t \in K$.  Using Property
(i) and {Taylor Decomposition}, we obtain
$g = g(t) + g'(t) (X - t) + c (X - t)^2$ for some $c \in K$. By
Property (iv), one has $g(t) = f(t)$. In addition, Property (iii)
yields $f(x) - g(x) \geq 0 = f(t) - g(t)$, for all $x \in K$, which
implies that $(f - g)'(t) = 0$ and $g'(t) = f'(t)$. By Property (ii),
the quadratic polynomial $g (X + t) = f(t) + f'(t) X + c X^2$ has at
most one root. This implies that the discriminant of $g$, namely
$f'(t)^2 - 4 c f(t)$ cannot be positive, thus one has
$c \geq \frac{f'(t)^2}{4 f(t)}$ (since $f(t) > 0$).

Finally, given a polynomial $g$ satisfying (i)-(iii) and (iv) for some
$t \in K$, one necessarily has $g = f_{t, c}$ with
$\frac{f'(t)^2}{4 f(t)} \leq c \in K$ and
$f_{t, c} = f(t) + f'(t) (X - t) + c (X - t)^2$.

In this case, one also has that the polynomial $g = f_{t, c'}$, with
$c' = \frac{f'(t)^2}{4 f(t)}$, fulfills (i)-(iii) and (iv). Indeed,
(i) and (iv) trivially hold. Let us prove that (ii) holds: when
$\deg f_{t, c'} = 0$, then $g = f(t) \geq 0$ and when
$\deg f_{t, c'} = 2$, then $g$ has a single root $\frac{-f'(t)}{2 c}$
and the minimum of $g$ is $g\left(\frac{-f'(t)}{2 c}\right) = 0$. The
inequalities $f_{t, c'} \leq f_{t, c} \leq f$ over $\R$ yield (iii).

Therefore, given $f \in K[X]$ with $f$ positive over $\R$, we are
looking for $t \in K$ such that the inequality $f \geq f_t$ holds over
$\R$, with
\[ 
f_t := f(t) + f'(t) (X - t) + \dfrac{f'(t)^2}{4 f(t)} (X - t)^2 \in K[X] \,.
\]
The main problem is to ensure that $t$ lies in $K$. If we choose $t$
to be a global minimizer of $f$, then $f_t$ would be the constant
polynomial $\min \{f(x) \mid x \in \R \}$. The idea is then to find
$t$ in the neighborhood of a global minimizer of $f$. The following
lemma shows that the inequality $f_t \leq f$ can be always satisfied
for $t$ in some neighborhood of a local minimizer of $f$.
\begin{lemma}
\label{th:ineqlocal}
Let $f \in \R[X]$. Let $a$ be a local minimizer of $f$ and suppose
that $f(a)>0$.  For all $t \in \R$ with $f(t)\neq 0$, let us
define the polynomial $f_t$:
\[
f_t := f(t) + f'(t) (X - t) + \dfrac{f'(t)^2}{4 f(t)} (X - t)^2 \in \R[X] \,.
\]
Then, there exists a neighborhood $U \subset \R$ of $a$ such
that the inequality $f_t(x) \leq f(x)$ holds for all
$(t, x) \in U \times U$.
\end{lemma}
\begin{proof}
  Set $n:=\deg f$.  It is easy to see that we can suppose without loss
  of generality that $a$ is the origin and that $f(0)=1$. Because of the Taylor formula
\[f=\sum_{k=0}^n\frac{f^{(k)}(t)}{k!}(X-t)^k,\]
we have 
\[
f - f_t=\sum_{k=2}^n\frac{f^{(k)}(t)}{k!}(X-t)^k-\dfrac{f'(t)^2}{4 f(t)}(X-t)^2=(X-t)^2\left(\sum_{k=2}^n\frac{f^{(k)}(t)}{k!}(X-t)^{k-2}-\frac{f'(t)^2}{4 f(t)}\right)
\]
for all $t\in\R$ with $f(t) \neq 0$. 
Let $h$ be the bivariate polynomial defined as follows:
\[h:=f(T)\left(\sum_{k=2}^n\frac{f^{(k)}(T)}{k!}(X-T)^{k-2}\right)-\frac14f'(T)^2\in\R[T,X].\]

Let us prove that $(a, a)$ is a local minimizer of $h$. 

Since  $f(0) = 1$, there exists $c \neq 0$, $\alpha \in \N$ and $g \in \R[X]$ such that $f - 1 = c X^\alpha + X^{\alpha+1} g$.
Therefore, $\lim_{x \to 0} \frac{f(x) - 1}{c x^\alpha} = 1$. 
Since $f - 1$ is non-negative over $\R$, one concludes that $c > 0$ and $\alpha$ is even.
Let us consider the lowest homogeneous part $H$ of $h$, that is the sum of all monomials of lowest degree involved in $h$. The lowest homogeneous part of $f'(T)^2$ is $c^2 (\alpha - 1)^2 T^{2 \alpha - 2}$ with degree $2 \alpha - 2$ while the lowest homogeneous part of $\sum_{k=2}^n\frac{f^{(k)}(T)}{k!}(X-T)^{k-2}$ is $c\sum_{k=2}^n\binom \alpha kT^{\alpha-k}(X-T)^{k-2}$ with degree $\alpha - 2$. 
Then
\[
H=c\sum_{k=2}^n\binom \alpha kT^{\alpha-k}(X-T)^{k-2} \,, 
\]
and thus
\[
(X-T)^2H=c((T+(X-T))^\alpha-T^n-nT^{\alpha-1}(X-T))=c(X^\alpha-\alpha
T^{\alpha-1}X+(\alpha-1)T^\alpha).
\]
%

%


%
Since $\lim_{\|(x,t)\| \to 0} \frac{h(x,t)}{H(x,t)} = 1$, it is enough to prove that $H$ is positive except at the origin in order to show that $(a,a) = (0,0)$ is a local minimizer of $h$.
Let us consider $(t,x)\in\R^2\setminus\{0\}$ and show that $H(t,x)>0$. If 
$t=x$, we have $H(t,x)=H(x,x)=\binom \alpha 2x^{\alpha-2}>0$. If $t\ne x$, then it is enough to show that
$(x-t)^2H(t,x)=c(x^\alpha-\alpha t^{\alpha-1}x+(\alpha-1)t^\alpha)>0$.
This is clear if $t=0$ since $c>0$ and $\alpha$ is even.  Now suppose that $t\ne0$ and define 
$\xi:=\frac xt\ne1$. Then one has 
$t^{-\alpha}(x-t)^2H(t,x)=c(\xi^\alpha-\alpha\xi+\alpha-1) > 0$.  
The positivity of $H$ follows from the fact that the univariate polynomial $r:=X^\alpha-\alpha X+\alpha-1$ is positive except at $1$ since $r'=\alpha X^{\alpha-1}-\alpha$. 
The positivity of $H$ implies that $(a, a)$ is a local minimizer of $h$. 

Let us define $q(X,T) := (X-T)^2 h$. 
Combining the fact that $(a, a)$ is a local minimizer of the two polynomials $h$, $(X-T)^2$ and the fact that $h(a, a) = f(a) f''(a) - \frac{1}{4} f'(a)^2 = 0$, we conclude that $(a, a)$ is also a local minimizer of $q$. 
Since $f(x) - f_t(x) = f(t) \, q(x,t)$, this yields the existence of a neighborhood $O \subset \R^2$ of $(a,a)$ such
that the inequality $f - f_t \geq 0$ holds for all
$(x, t) \in O$. Since there exists some neighborhood $U \subset \R$ of $a$, such that the rectangle $U \times U$ is included in $O$, this proves the initial claim.

\end{proof}
Lemma~\ref{th:ineqlocal} states the existence of a neighborhood $U$ of
a local minimizer of $f$ such that the inequality $f_t(x) \leq f(x)$
holds for all $(x, t) \in U \times U$ . Now, we show that with such a neighborhood $U$ of the smallest global minimizer of
$f$, the inequality $f_t(x) \leq f(x)$ holds for all
$t \in U$ and for all $x \in \R$.
\begin{proposition}
\label{th:ineqglobal}
Let $f \in \R[X]$ with $\deg f > 0$. Assume that $f$ is positive over
$\R$. Then, there exists a smallest global minimizer $a$ of $f$ and a
positive $\epsilon \in \R$ such that for all $t \in \R$ with
$a -\epsilon < t < a$, the quadratic polynomial $f_t$, defined by
\[
f_t := f(t) + f'(t) (X - T) + \dfrac{f'(t)^2}{4 f(t)} (X - T)^2 \in \R[X] \,,
\]
satisfies $f_t \leq f$ over $\R$.
\end{proposition}
\begin{proof}
  The existence of $a$ is straightforward.  First, we handle the case
  when $\deg f = 2$. Using Taylor Decomposition of $f$ at $t$, one
  obtains $f = f(t) + f'(t) (X - T) + \frac{f''(t)}{2}(X-T)^2$. Since
  $f$ has no real root, the 
    discriminant of $f$ is negative, namely
  $f'(t)^2 - 4 f(t) \frac{f''(t)}{2} < 0$.  It implies that
  $\frac{f'(t)^2}{4 f(t)} < \frac{f''(t)^2}{2}$, ensuring that the
  inequality $f_t \leq f$ holds over $\R$.

  In the sequel, we assume that $\deg f > 2$. We can find a
  neighborhood $U$ as in Lemma~\ref{th:ineqlocal} and without loss of
  generality, let us suppose that
  $U = [a - \epsilon_0, a + \epsilon_0]$ for some positive
  $\epsilon_0$, so that $f'$ has no real root in
  $[a - \epsilon_0, a)$. Then, the inequality $f_t(x) \leq f(x)$ holds
  for all $x, t \in U$. Next, we write
  $f - f_t = \sum_{i=0}^n a_{i t} x^i$, with $a_{i t} \in \R$ and
  $n = \deg f > 2$ and define the following function:
\[
U \to \R : t \mapsto C_t := \max \left\{1, \dfrac{|a_{0 t}|}{|a_{n t}|}, \dots, \dfrac{|a_{(n-1) t}|}{|a_{n t}|}  \right\} \,.
\]
Note that the Cauchy bound (Lemma~\ref{th:cauchy}) implies that for
all $t \in U$, all real roots of $f - f_t$ lie in $[-C_t, C_t]$.  In
addition, the closed interval domain $U$ is compact, implying that the
range values of the function $U \to \R : t \mapsto C_t$ are
bounded. Let $C \in \R$ with $C \geq C_t$ for all $t \in U$. Then, for
all $t \in U$, all real roots of $f - f_t$ lie in the interval
$[-C, C]$ and we can assume without loss of generality that
$-C < a - \epsilon_0 < a < a + \epsilon_0 < C$. Let us define
$M := \min \{ f(x) \mid x \in [-C, a - \epsilon_0] \}$. By definition,
$a$ is the global minimizer of $f$, ensuring that $f(a) < M$.  For all
$t \in [a - \epsilon_0, a)$, the quadratic polynomial $f_t$ has one
real root $N_t := \frac{-2f(t)}{f'(t)} + t$. When
$t \in [a - \epsilon_0, a)$ converges to $a$, then $f'(t) < 0$
converges towards 0 and $-2 f(t)$ converges towards $-2 f(a) < 0$.
Thus, the corresponding limit of $N_t$ is $+ \infty$. In addition,
$f_t(-C)$ tends to $f_a(-C) = f(a) < M$. Therefore, there exists some
$\epsilon \in (0, \epsilon_0]$ such that for all
$t \in (a - \epsilon, a)$, one has $N_t \in [C, \infty)$ and
$f_t(-C) < M$. For all $t \in (a - \epsilon, a)$, we partition $\R$
into five intervals and prove that the inequality $f_t \leq f$ holds
on each interval:
\begin{itemize}
\item The inequality $f_t \leq f$ holds over $(- \infty, -C]$: it
  comes from the fact that $f_t(-C) < M \leq f(-C)$ and the fact
  $f - f_t$ has no real root in $(- \infty, -C]$.
\item The inequality $f_t \leq f$ holds over $(- C, a - \epsilon_0]$:
  $f_t$ is monotonically decreasing over $(-\infty, N_t]$. Since one
  has $-C < a - \epsilon_0 < C \leq N_t $, then $f_t$ is monotonically
  decreasing over $(- C, a - \epsilon_0]$. This implies that for all
  $x \in (- C, a - \epsilon_0]$, one has
  $f_t(x) \leq f_t(-C) < M \leq f(x)$.
\item The inequality $f_t \leq f$ holds over $[a - \epsilon_0, a)$: it
  follows from the fact that $[a - \epsilon_0, a) \subseteq U$.
\item The inequality $f_t \leq f$ holds over $[a, C)$: $f_t$ is
  monotonically decreasing over $(-\infty, N_t]$. Since one has
  $a < C \leq N_t$, then $f_t$ is monotonically decreasing over
  $[a, C)$. Since $a$ is a global minimizer of $f$ and $a \in U$, one
  has $f_t(x) \leq f_t(a) \leq f(a) \leq f(x)$ for all $x \in [a, C)$
  .
\item The inequality $f_t \leq f$ holds over $[C, \infty]$: the claim
  is implied by the fact that $f_t(N_t) = 0 < f(N_t)$,
  $N_t \in [C, \infty]$ together with the fact that $f -f_t$ has no
  real root in $[C, \infty]$.
\end{itemize}
\end{proof}
\if{
We can now state our \markus{Nichtnegativstellensatz}:
\markus{C'est un peu étrange de formuler cela. D'abord ça a du être connu depuis longtemps
sans doute. Comme j'explique dans mon Diplomarbeit, l'algorithme \texttt{univsos1} ne marche que 
pour les corps ordonnées archimédien, c.-à-d. pour les sous corps de $\R$. Par contre, en utilisant
le principe de Tarski, on peut appliquer la même preuve (pas l'algorithme) toujours pour des corps
ordonnées qui sont dense dans leur clôture réelle. Mais ce n'est pas intéressant: De toute façon il
y a une preuve complètement non-constructif pour toutes les corps ordonnées de l'existence de
la décomposition en somme de carrés pondérée qui se base sur la solution du problème 17ième
de Hilbert par Artin combiné avec le principe de Cassels pour éliminer les racines. Peut-être qu'il
y même d'autre preuves. Ici je dirais qu'on pourrait facilement se restreindre sur les sous corps
de $\R$ (avec l'ordre induit par $\R$). Je ne vois pas l'intérêt de parler des autres choses (en tout
cas si on se décide de parler d'autre corps ordonnées il faudrait d'abord faire une recherche de
la literature pour savoir ce qu'il faut citer. Mais si on se restreint aux sous corps de $\R$, c'est irritant
de formuler cette proposition parce que automatiquement la question se pose si ça ne reste pas
vrai en plus grand généralité. Quoi faire?}
}\fi
%

\begin{proposition}
\label{th:nnsatz}
Let $K$ be {a subfield of $\R$} and $f \in K[X]$ with
$\deg f = n \geq 1$. Then $f$ is non-negative on $\R$ if and only if
$f$ is a weighted sum of $n$ polynomial squares in $K[X]$,~i.e.~there
exist $a_1, \dots, a_n \in K^{\geq 0}$ and $g_1, \dots, g_n \in K[X]$
such that $f = \sum_{i=0}^n a_i g_i^2$.
\end{proposition}
\begin{proof}
  The {\em if} part is straightforward. For the other direction,
  assume that $f$ is non-negative on $\R$ and $n$ is even. The proof is
  by induction over $n$. The base case $n=2$ follows from
  Lemma~\ref{th:parab}.  For the induction case, let us consider
  $n \geq 4$.

When $f$ is not a square-free polynomial, we show that $f$ is a
weighted sum of $n - 2$ polynomial squares. We can write $f = g h^2$,
for some polynomials $g, h \in K[X]$ with $\deg g \leq \deg f -
2$.
This gives $g(x) = \frac{f(x)}{h(x)^2} \geq 0$ for all $x \in \R$ such
that $h(x) \neq 0$. Since $h$ has a finite number of real roots, $g$
is non-negative on $\R$. Using the induction hypothesis, $g$ is a
weighted sum of $n - 2$ polynomial squares. Therefore, $f$ is also a
weighted sum of $n - 2$ polynomial squares.

When $f$ is a square-free polynomial, then $f$ has no real root, which
implies by Lemma~\ref{th:sfextend} that $f$ is neither a square-free
polynomial in $K[X]$ nor in $\R[X]$. Thus, $f$ is positive on
$\R$. Using Proposition~\ref{th:ineqglobal}, there exists some
$t \in K$ ($K$ is dense in $\R$) and a quadratic polynomial
$f_t \in K[X]$ such that the inequalities $0 \leq f_t(x) \leq f(x)$
holds for all $x \in \R$ and $f_t(t) = f(t)$. The polynomial $f - f_t$
has degree $n$, takes only non-negative values. In addition
$(f - f_t)(t) = 0$, thus $f - f_t$ is not a square-free
polynomial. Hence, we are in the above case, implying that $f - f_t$
is a weighted sum of $n - 2$ polynomial squares. From
Lemma~\ref{th:parab}, $f_t$ is a weighted sum of 2 polynomial squares,
implying that $f$ is a weighted sum of $n$ polynomial squares, as
requested.
\end{proof}
\subsection{Algorithm~\texttt{univsos1}}
\label{sec:univsos1}
The global minimizer $a$ is a real root of $ f' \in K[X]$.
{Therefore, by using root isolation
  techniques~\cite[Chap. 10]{BPR06}, one can isolate all the real
  roots of $f'$ in distinct intervals with bounds in $K$. Such
  techniques rely on applying successive bisections, so that one can
  arbitrarily reduce the width of every interval and sort them
  w.r.t.~their lower bounds.}
%
Eventually, we apply this procedure to find a sequence of elements in
$K$ converging from below to the smallest global minimizer of $f$ in
order to find a suitable $t$. We denote by $\texttt{parab}(f)$ the
corresponding procedure which returns the polynomial
$f_t := \frac{f'(t)^2}{f(t)}(X-t)^2 + f'(t)(X - t) + f(t)$ such that
$t \in K$ and $f \geq f_t$ over $\R$.

Algorithm~$\univsosone$, depicted in Figure~\ref{alg:univsos1}, takes
as input a polynomial $f \in K[X]$ of even degree
$n \geq 2$. The steps performed by this algorithm correspond to what
is described in the proof of Proposition~\ref{th:nnsatz} and relies on
two auxiliary procedures. The first one is the procedure
\texttt{parab} performing root isolation (see
Step~\lineref{line:parab}). The second one is denoted by
\texttt{sqrfree} and performs square-free decomposition: for a given
polynomial $f \in K[X]$, $\texttt{sqrfree}(f)$ returns two polynomials
$g$ and $h$ in $K[X]$ such that $f = g h^2$. When $f$ is square-free,
the procedure returns $g = f$ and $h = 1$ (in this case $\deg h = 0$).
As in the proof of Proposition~\ref{th:nnsatz}, this square-free
decomposition procedure is performed either on the input polynomial
$f$ (Step~\lineref{line:sqrfree1}) or on the non-negative polynomial
$(f - f_t)$ (Step~\lineref{line:sqrfree2}).
\begin{figure}[!t]
\begin{algorithmic}[1]                    
\Require non-negative polynomial $f \in K[X]$ of degree $n\geq 2$, with $K$ a subfield of $\R$
\Ensure pair of lists of polynomials $(\hlist, \qlist)$ with coefficients in $K$
\State $\hlist \gets [\,]$, $\qlist \gets [\,]$. \label{line:list}
\While {$\deg f > 2$}
\State $(g, h) := \texttt{sqrfree}(f)$ \label{line:sqrfree1} \Comment{$f = g h^2$}
\If {$\deg h > 0$}  $\hlist \gets  \hlist \cup \{ h \}$, $\qlist \gets  \qlist \cup \{ 0 \}$, $f \gets g$
\Else \State $f_t := \texttt{parab}(f)$ \label{line:parab} 
\State $(g, h) := \texttt{sqrfree}(f - f_t)$ \label{line:sqrfree2} 
\State $\hlist \gets  \hlist \cup \{ h \}$, $\qlist \gets  \qlist \cup \{ f_t \}$, $f \gets g$
\EndIf
\EndWhile
\State $\hlist \gets  \hlist \cup \{ 0 \}$, $\qlist \gets  \qlist \cup \{ f \}$
\State \Return $\hlist$, $\qlist$
\end{algorithmic}
\caption{$\univsosone$: algorithm to compute SOS decompositions of non-negative univariate polynomials.}
\label{alg:univsos1}
\end{figure}
The output of Algorithm~$\univsosone$ is a pair of lists of
polynomials in $K[X]$, allowing to retrieve an SOS decomposition of
$f$. By Proposition~\ref{th:nnsatz} the length of all output lists,
denoted by $r$, is bounded by $n/2$.  If we note $h_r, \dots, h_1$ the
polynomials belonging to $\hlist$ and $q_r, \dots, q_1$ the positive
quadratic polynomials belonging to $\qlist$, one obtains the following
Horner-like decomposition:
$f = h_r^2 \bigl( h_{r-1}^2 ( h_{r-2}^2 (\dots) + q_{r-2} ) + q_{r-1}
\bigr) + q_r$.
Thus, each positive quadratic polynomial $q_i$ being a weighted SOS
polynomial, this yields a valid weighted SOS decomposition for $f$.
\begin{example}
\label{ex:1}
Let us consider the polynomial $f := \frac{1}{16} X^6 + X^4 - \frac{1}{9} X^3 - \frac{11}{10} X^2 + \frac{2}{15} X + 2 \in \Q[X]$. 

We describe the different steps performed by Algorithm~$\univsos1$: 
\begin{itemize}
\item The polynomial $f$ is square-free and the algorithm starts by providing the value $t = -1$ as an approximation of the smallest minimizer of $f$. With $f(t) = \frac{1397}{720}$ and $f'(t) = \frac{-19}{8}$, one obtains $f_{-1} = \frac{720}{1397} (-\frac{19}{16} X + \frac{271}{360} )^2$. \item Next, after obtaining the square-free decomposition $f(X) - f_{-1} = (X+1)^2  g$, the same procedure is applied on $g$. One obtains the value $t=1$ as an approximation of the smallest minimizer of $g$ and $g_{1} = \frac{502920}{237293}(-\frac{1}{18} X + \frac{88411}{167640})^2$. 
\item Eventually, one obtains the square-free decomposition $g(X) - g_{1} = (X-1)^2 h$ with $h = \frac{1}{16}(X - \frac{19108973}{17085096})$.
\end{itemize}
Overall, Algorithm~$\univsos1$ provides the lists $\hlist = [1,X+1,1,X-1,0]$ and $\qlist = 
[
\frac{720}{1397} (-\frac{19}{16} X + \frac{271}{360} )^2,  
0,
\frac{502920}{237293}(-\frac{1}{18} X + \frac{88411}{167640})^2,
0,
\frac{1}{16}(X - \frac{19108973}{17085096})
]$, yielding the following weighted SOS decomposition: 
\[
f :=  \Biggl(  (X+1)^2  \biggl((X-1)^2 \Bigl(\frac{1}{16}(X - \frac{19108973}{17085096})^2
\Bigr) + \frac{502920}{237293}(-\frac{1}{18} X + \frac{88411}{167640})^2 \biggr)  \Biggr)  + \frac{720}{1397} (-\frac{19}{16} X + \frac{271}{360} )^2
 .
\]
\end{example}
In the sequel, we analyze the complexity of Algorithm~$\univsosone$ in the particular case $K = \Q$. We provide bounds on the bitsize of related SOS decompositions as well as bounds on the arithmetic cost required for computation and verification.
\subsection{Bit size of the output}
\label{sec:bitsize1}
%

%
\begin{lemma}
\label{th:bitroot}
Let $f \in \Z[X]$ be a positive polynomial over $\R$, with
$\deg f = n$ 
and $\tau$ be an upper bound on the bitsize of the
coefficients of $f$. 
When applying Algorithm~$\univsosone$ to $f$, the
sub-procedure \texttt{parab} outputs a polynomial $f_t$ such that
$\tau(t) = \bigo{(n^2 \tau)}$.
\end{lemma}
\begin{proof}
  \if{ Let us write the polynomial
    $f' = \sum_{i=0}^{n-1} \frac{b_i}{c_i} X^i$, such that each
    rational coefficient $\frac{b_i}{c_i}$ is irreducible and let
    $l \in \Z$ be the least common multiples of all the $c_i$.  First,
    let us estimate the bitsize of the coefficients of the polynomial
    $l f' \in \Z[X]$. One has
    $\tau(l) \leq \sum_{i=0}^{n-1} \tau(c_i) \leq n \tau (f')$,
    yielding
    $\tau (l f') \leq (n+1) \tau (f') \leq (n+1) \log_2 n + (n+1) \tau
    (f)$.  }\fi 

  Let us consider the set $S\subseteq \Q$ defined by:
\[ 
S := 
\{ 
t \in \Q \mid
\forall x \in \R \,, f(t)^2 + f'(t) f(t) (x - t) + f'(t)^2 (x - t)^2  \leq 4 f(t) f(x) \,
\}
\,.
\]
The polynomial involved in $S$ has degree $2 n$, with maximum bitsize
of the coefficients upper bounded by $2 \tau$.  {Thanks to the
  complexity analysis of the quantifier elimination procedure
  described in~\cite[\S11.1.1]{BPR06} the set $S$ can be described by
  polynomials with maximum bitsize coefficients upper bounded by
  $\bigo{(n^2 \tau)}$.}  Since $t$ is a rational root of one of these
polynomials, the rational zero
theorem~\cite{swokowski1989fundamentals} implies that
$\tau(t) = \bigo{(n^2 \tau)}$.
%
\end{proof}
\begin{lemma}
\label{th:bitft}
Let $f \in \Z[X]$ be a positive polynomial over $\R$, with
$\deg f = n$ and $\tau$ be an upper bound on the bitsize of the
coefficients of $f$.  When applying Algorithm~$\univsosone$ to $f$, the
sub-procedure \texttt{parab} outputs a polynomial $f_t$ such that
$\tau(f_t) = \bigo{(n^3 \tau)}$. Moreover, there exist polynomials $\hat{f}, \hat{f_t}, g \in \Z[X]$ such that $\hat{f} - \hat{f_t} = (X - t)^2 g$ and $\tau(g) = \bigo{(n^3 \tau)}$.
\end{lemma}
\begin{proof}
One can write $f_t = M_2(t) X^2 + M_1(t) X + M_0(t)$ with  
\begin{align*}
M_2(t) & := \frac{f'(t)^2}{4f(t)} \,,\\
M_1(t) & := \frac{2 f'(t) (2 f(t) - t f'(t))}{4 f(t)}  \,,\\
M_0(t) & := \frac{(2 f(t) - t f'(t))^2}{4 f(t)} \,,
\end{align*}
and $\| f_t \|_{\infty} \leq \max \{ M_2(t), |M_1(t)|, M_0(t) \}$.
One has $0 \leq M_0(t) = f_t(0) \leq f(0) \leq \| f \|_{\infty}$.

In addition,
$0 \leq M_0(t) + M_1(t) + M_2(t) = f_t(1) \leq f(1) \leq (n+1) \| f
\|_{\infty}$
and
$0 \leq M_0(t) - M_1(t) + M_2(t) = f_t(-1) \leq f(-1) \leq (n+1) \| f
\|_{\infty}$.
Thus, one has
$M_0(t) + |M_1(t)| + M_2(t) \leq (n+1) \| f \|_{\infty}$, wich implies
that $\| f_t \|_{\infty} \leq (n+1) \| f \|_{\infty}$.

Now let us note $t = \frac{t_1}{t_2}$, with
$t_1 \in \Z, t_2 \in \Z \backslash \{0\}$, {$t_1$ and $t_2$ being coprime.} Let us define the polynomials
$\hat{f}(X) := t_2^{2 n} f(t) f(X)$ and
$\hat{f_t}(X) := t_2^{2 n} f(t) f_t(X)$. By writing
$f(X) = \sum_{i=0}^n a_i X^i$, one has
$t_2^{2 n} f(t) = \sum_{i=0}^n a_i t_1^i t_2^{2n-i} \leq \| f
\|_{\infty} |t_1|^i |t_2|^{2 n - i}$
and $\tau(\hat{f}) \leq \tau + \tau(t^{2n})$. By
Lemma~\ref{th:bitroot}, one has $\tau(\hat{f}) = \bigo{(n^3 \tau)}$.

The polynomials $\hat{f}(X), \hat{f_t}(X)$ are polynomials in $\Z[X]$
and since
$\| \hat{f_t} \|_{\infty} \leq (n+1) \| \hat{f} \|_{\infty}$, the
triangular inequality
$\| \hat{f} - \hat{f_t} \|_{\infty} \leq \| \hat{f} \|_{\infty} + \|
\hat{f_t} \|_{\infty}$
implies that
$\tau(\hat{f} - \hat{f_t}) \leq \log_2 (n+2) + \tau(\hat{f})$. In
addition,
$\tau(f_t) \leq \tau(\hat{f_t}) + \tau(t_2^{2n} f(t)) = \bigo{(n^3
  \tau)}$.

As in the proof of {Proposition}~\ref{th:nnsatz}, one has $(\hat{f} - \hat{f_t})(t) = 0$ which allows to write the square-free decomposition of the polynomial $\hat{f} - \hat{f_t} \in \Z[X]$ as $\hat{f} - \hat{f_t} = (X - t)^2 g$, with $g \in \Z[X]$. By Lemma~\ref{th:bitdiv}, one has $\tau(g) \leq n-2 + \tau(\hat{f} - \hat{f_t}) + \log_2(n + 1) \leq n-2 + 2 \log_2(n + 2) + \tau(\hat{f}) = \bigo{(n^3 \tau)}$, which concludes the proof.
\end{proof}
\begin{theorem}
\label{th:bitsize}
Let $f \in \Z[X]$ be a positive polynomial over $\R$, with
$\deg f = n = 2 k$ and $\tau$ be an upper bound on the bitsize of the
coefficients of $f$.  Then the maximum bitsize of the coefficients involved in the
SOS decomposition of $f$ obtained with Algorithm~$\univsosone$ is
upper bounded by
$ \bigo{( (k!)^3 \tau )} = \bigo{((\frac{n}{2})^{\frac{3n}{2}}
  \tau)}$.
\end{theorem}
\begin{proof}
  \if{ With $k = n/2$, one writes the polynomial
    $f = \sum_{i=0}^{2 k} \frac{b_i}{c_i} X^i$, such that each
    rational coefficient $\frac{b_i}{c_i}$ is irreducible and let
    $l \in \Z$ be the least common multiples of all the $c_i$.  First,
    let us estimate the bitsize of the coefficients of the polynomial
    $l f \in \Z[X]$. One has
    $\tau(l) \leq \sum_{i=0}^{2 k} \tau(c_i) \leq (2k+1) \tau $,
    yielding
    $\tau (l f) \leq \tau(l) + \tau (f) \leq (2k+2) \tau = \bigo{(k
      \tau)}$.
  }\fi With $k = n/2$ and starting from the polynomial $f$,
  Algorithm~$\univsosone$ generates, in the worst case scenario, two
  sequences of polynomials
  $f_{k}, \dots, f_1 \in \Z[X], q_k, \dots, q_{2} \in \Z[X]$ as well
  as rational numbers $t_k, \dots, t_2 \in \Q$ such that $f_k = f$,
  $t_i = \frac{t_{i1}}{t_{i2}}$, with $t_{i1} \in \Z$,
  $t_{i2} \in \Z \backslash \{0\}$ and
\begin{equation}
\label{eq:univsos}
t_{i2}^{4 i} f_i(t_i) f_i - q_i = (X - t_i)^2 f_{i-1}  \,, \quad i=2,\dots,k   \,.
\end{equation}
From Lemma~\ref{th:bitft}, for all $i=2,\dots,k$, one has
$\tau(f_{i-1}) = \bigo{(i^3 \tau (f_i))}$. This yields
$\tau(f_1) = \bigo{\bigl( (k!)^3 \tau(f) \bigr)}$.
%

Using Stirling's formula, we obtain
$k! \leq 2 \sqrt{2 \pi k} \bigl(\frac{k}{e}\bigr)^k$ and
$ (k!)^3 \leq 1024 \sqrt{2} \pi^{\frac{3}{2}} k^{\frac{3}{2}}
\bigl(\frac{k}{e}\bigr)^{3k}$, 
where $e$ denotes the Euler number.  Since $k \leq e^{k}$ for each
integer $k \geq 1$ and $\frac{3}{2} < 3$, one has
$(k!)^3 \in \bigo{(k^{3k})}$, yielding
$\tau(f_i) = \bigo{((\frac{n}{2})^{ \frac{3n}{2} } \tau) }$, for all
$i = 1,\dots,k$.  Similarly, we obtain
$\tau(q_i) = \bigo{((\frac{n}{2})^{ \frac{3n}{2} } \tau) }$, for all
$i = 1,\dots,k$.
Finally, using Lemma~\ref{th:bitroot}, one has $\tau(t_i) = \bigo{(i^2 \tau(f_i))}$, yielding the desired result. 
\end{proof}
\subsection{Bit complexity analysis}
\label{sec:bitop1}
\begin{theorem}
\label{th:bitop}
Let $f \in \Z[X]$ be a positive polynomial over $\R$, with
$\deg f = n = 2 k$ and $\tau$ be an upper bound on the bitsize of the
coefficients of $f$.
Then, on input $f$, Algorithm~$\univsosone$ runs in boolean time
$$ \bigotilde{( k^3 \cdot (k!)^3 \tau )} =
\bigotilde{\left (\left (\frac{n}{2}\right )^{\frac{3 n}{2}}
    \tau\right )}.$$
\end{theorem}
\begin{proof}
  For $i=2, \ldots, k$ we obtain each polynomial $f_{i-1}$ as in the
  proof of Theorem~\ref{th:bitsize} by computing the square-free
  decomposition of the polynomial $t_{i2}^{4 i} f_i(t_i) f_i - q_i$.
%
  It follows by Lemma~\ref{th:costsqrfree} that the polynomial $f_{i-1}$ can be computed using an expected number
  of $\bigotilde{(i^2 \cdot i^3 \tau(f_i))}$ boolean operations. The number
  of boolean operations to compute all polynomials $f_1, \dots, f_{k-1}$
  is thus bounded by
  $\bigotilde{\bigl(k^2 \cdot k^3 \tau(lf) + (k-1)^2 (k-1)^3 k^3 \tau(lf)
    + \dots + (k!)^3 \tau(lf) \bigr)}$. 

{    
%
%
  For each $i=2, \dots, k$, the bitsize of the rational number $t_i$
  is upper bounded by $\bigo{(i^2 \tau(f_i))}$. Therefore, $t_i$ can
  be computed by approximating the roots of $f_i'$ with isolating
  intervals of radius less than $2^{-i^2 \tau(f_i)}$. By
  Lemma~\ref{th:rootisolation}, the corresponding computation cost is
  $\bigotilde{(i^3 \tau(f_i))}$ boolean operations. The number of
  boolean operations to compute all rational numbers $t_2, \dots, t_k$
  is bounded by
  $\bigotilde{\bigl(k^3 \cdot k^3 \tau(lf) + (k-1)^{3} (k-1)^3 k^3
    \tau(lf) + \dots + (k!)^3 \tau(lf) \bigr)}$.}
    
{In addition, one has
  $k^3 \cdot k^3 + (k-1)^3 (k - 1)^3 k^3 + \dots + (k!)^3 = (k!)^3
  \sum_{i=1}^{k} \frac{1}{(i!)^3} \leq 2 k^3 \cdot (k!)^3$.
  Using Stirling's formula, we obtain
  $k^3 \cdot (k!)^3 \leq 1024 \sqrt{2} \pi^{\frac{3}{2}}
  k^{\frac{9}{2}} \bigl(\frac{k}{e}\bigr)^{3k}$.
  Since $k^{3/2} \leq e^{k}$ for each integer $k \geq 1$, we obtain
  the announced complexity.}
\end{proof}
{  
For a given polynomial $f$ of degree $2 k$, one can check  the correctness of the SOS decomposition obtained with Algorithm~$\univsosone$ by evaluating this SOS polynomial at $2 k + 1$ distinct points and compare the results with the ones obtained while evaluating $f$ at the same points.
}  
\begin{theorem}
\label{th:verif}
Let $f \in \Z[X]$ be a positive polynomial over $\R$, with
$\deg f = n = 2 k$ and $\tau$ be an upper bound on the bitsize of the
coefficients of $f$.  Then one can check the correctness of the SOS
decomposition of $f$ obtained with Algorithm~$\univsosone$ within 
\[ \bigotilde{( k \cdot (k!)^3 \tau )} =
\bigotilde{((\frac{n}{2})^{\frac{3 n}{2}} \tau)}\] 
boolean operations.
\end{theorem}
\begin{proof}
  From~\cite[Corollary 8.27]{Gathen99}, the cost of polynomial
  multiplication in $\Z[X]$ of degree less than $n = 2k$ with
  coefficients of bitsize upper bounded by $B$ is bounded by
  $\bigotilde{(k \cdot B)}$.  By Theorem~\ref{th:bitsize}, the maximal
  bitsize of the coefficients of the SOS decomposition of $f$ obtained
  with Algorithm~$\univsosone$ is upper bounded by
  $B = \bigo{( (k!)^3 \tau )}$.  Let us consider $2 k+1$ distinct integers, with maximal bitsize upper bounded by $\log_2 n$. 
  Therefore, from~\cite[Corollary
  10.8]{Gathen99}, the cost of the evaluation of this decomposition at the $2 k + 1$ points 
  can be performed using at most
  $\bigotilde{(k \cdot (k!)^3 \tau)}$ boolean operations, the desired
  result.
\end{proof}
%
\begin{remark}
\label{rk:univsos1}
Let $f_k = f \in \Z[X]$. Under the strong assumption that all polynomials $f_k, \dots, f_1$ involved in Algorithm $\univsosone$ have at least one integer global minimizer, then Algorithm $\univsosone$ has a polynomial complexity. Indeed, in this case, $q_i = f_i(t_i)$, $\tau(t_i) = \bigo{(\tau(f_i))}$ and $\tau(f_{i-1}) = \bigo{( 2 (i - 1) + \tau(f_i)  )}$, for all $i=2,\dots,k$. {Hence, the maximal bitsize of the coefficients involved in the SOS decomposition of $f$ is upper bounded by $\bigo{(k^2 + \tau)}$ and this decomposition can be computed using an expected number of $\bigotilde{(k^4 + k^3 \tau)}$ boolean operations.} 
\end{remark}

\section{Nichtnegativstellens\"atze with perturbed polynomials}
\label{sec:method2}
Here, we recall the algorithm  given
in~\cite[Section~5.2]{Chevillard11}. The description of this
algorithm, denoted by~$\univsostwo$, is given in
Figure~\ref{alg:univsos2}.

\subsection{Algorithm~\texttt{univsos2}}
\label{sec:univsos2}
Given a subfield $K$ of $\R$ and a non-negative polynomial
$f = \sum_{i=0}^n f_i \, X^i \in K[X]$ of degree $n = 2 k$, one first
obtains the square-free decomposition of $f$, yielding $f = p \, h^2$
with $p > 0$ on $\R$ (see Step~\lineref{line:sqrfree3}).Then the idea
is to find a positive number $\varepsilon > 0$ in $K$ such that the
perturbed polynomial
$p_\varepsilon(X) := p(X) - \varepsilon \sum_{i=0}^k X^{2 i}$ is also
positive on $\R$ (see~\cite[Section~5.2.2]{Chevillard11} for more
details). This number is computed thanks to the loop going from
Step~\lineref{line:epsi} to Step~\lineref{line:epsf} and relies on the
auxiliary procedure $\hasrealroots$ which checks whether the
polynomial $p_\varepsilon$ has real roots using {root isolation techniques}. As mentioned in~\cite[Section~5.2.2]{Chevillard11}, the
number $\varepsilon$ is divided by 2 again to allow a margin of safety
(Step~\lineref{line:epsafe}).

Note that one can always ensure that the leading coefficient
${\ell} := p_n$ of $p$ is the same as the leading coefficients $f_n$ of the
input polynomial $f$.

We obtain an approximate rational sums of squares decomposition of the
polynomial $p_\varepsilon$ with the auxiliary procedure
$\sumtwosquares$ (Step~\lineref{line:sumtwosquares}), relying on an
arbitrary precision complex root finder.  Recalling
Theorem~\ref{th:sosR}, this implies that the polynomial $p$ can be
approximated as close as desired by the weighted sum of two polynomial
squares in $\Q[X]$, that is
${\ell} s_1^2+ {\ell} s_2^2$.

Thus there exists a remainder polynomial
$u := p_\varepsilon - \\
{\ell} s_1^2 - {\ell} s_2^2$ with coefficients of
arbitrary small magnitude (as mentioned
in~\cite[Section~5.2.3]{Chevillard11}).  The magnitude of the
coefficients converges to 0 as the precision $\delta$ of the complex
root finder goes to infinity. The precision is increased thanks to the
loop going from Step~\lineref{line:deltai} to
Step~\lineref{line:deltaf} until a condition between the coefficients
of $u$ and $\varepsilon$ becomes true, ensuring that
$\varepsilon \sum_{i=0}^k X^{2 i} + u(X)$ also admits a weighted SOS
decomposition. For more details,
see~\cite[Section~5.2.4]{Chevillard11}.

\begin{figure}[!t]
\begin{algorithmic}[1]                    
\Require non-negative polynomial $f \in K[X]$ of degree $n \geq 2$, with $K$ a subfield of $\R$, $\varepsilon \in K$ such that $0 < \varepsilon < f_n$, precision $\delta \in \N$ for complex root isolation
\Ensure list $\clist$ of numbers in $K$ and list $\slist$ of polynomials in $K[X]$
\State $(p, h) \gets \texttt{sqrfree}(f)$ \label{line:sqrfree3}
\Comment{$f = p \, h^2$}
\State $n' := \deg p$, $k := n'/2$
\State $p_\varepsilon \gets p - \varepsilon \sum_{i=0}^k X^{2 i}$
\While {$\hasrealrootsfun{p_\varepsilon}$} \label{line:epsi}
\State $\varepsilon \gets \frac{\varepsilon}{2}$, $p_\varepsilon \gets p - \varepsilon \sum_{i=0}^k X^{2 i}$
\EndWhile \label{line:epsf}
\State $\varepsilon \gets \frac{\varepsilon}{2}$ \label{line:epsafe}
\State $(s_1, s_2) \gets \sumtwosquaresfun{p_\varepsilon}{ \delta}$ \label{line:sumtwosquares}
\State ${\ell} \gets f_n$, $u \gets p_\varepsilon - {\ell} s_1^2 - {\ell} s_2^2$, 
 $u_{-1} \gets 0$, $u_{2k + 1} \gets 0$
\Comment{$u = \sum_{i=0}^{2 k - 1} u_i X^i$}
\While {$ \varepsilon < \min_{0 \leq i \leq k} \bigl\{ \, \frac{|u_{2 i + 1}|}{4} 
- u_{2 i} + |u_{2 i - 1}| \, \bigr\} $} \label{line:deltai}
\State $\delta \gets 2 \delta$, $(s_1, s_2) \gets \sumtwosquaresfun{p_\varepsilon}{\delta}$, $u \gets p_\varepsilon - {\ell} s_1^2 - {\ell} s_2^2$
\EndWhile \label{line:deltaf}
\State $\clist \gets [{\ell}, {\ell}] $, $\slist \gets [h \, s_1, h \, s_2] $
\For {$i=0$ to $k-1$}
\State $\clist \gets  \clist \cup \{ |u_{2 i + 1}| \} $, $\slist \gets  \slist \cup 
\{ h \, (X^{i + 1} + 
\frac{\sgn{(u_{2 i + 1})}}{2} X^{i}) \}$
\State $\clist \gets  \clist \cup \{ \varepsilon - \frac{|u_{2 i + 1}|}{4} 
+ u_{2 i} - |u_{2 i - 1 }| \} $, $\slist \gets  \slist \cup 
\{ h \, X^i \}$
\EndFor
\State \Return $\clist \cup \{ \varepsilon + u_{n} - |u_{n - 1}| \} $, $\slist \cup 
\{h \, X^k \}$
\end{algorithmic}
\caption{$\univsostwo$: algorithm to compute SOS decompositions of non-negative univariate polynomials.}
\label{alg:univsos2}
\end{figure}

The reason why Algorithm~$\univsostwo$ terminates is the following: at
first, one can always find a sufficiently small perturbation
$\varepsilon$ such that the perturbed polynomial $p_\varepsilon$
remains positive. Next, one can always find sufficiently precise
approximations of the complex roots of $p_\varepsilon$ ensuring that
the error between the initial polynomial $p$ and the approximate SOS
decomposition is compensated thanks to the perturbation term.

The output of Algorithm~$\univsostwo$ are a list of numbers in $K$ and
a list of polynomials in $K[X]$, allowing to retrieve a weighted SOS
decomposition of $f$. The size $r$ of both lists is equal to
$2 k+3 = n' + 3 \leq n + 3$.  If we note $c_r, \dots, c_1$ the numbers
belonging to $\clist$ and $s_r, \dots, s_1$ the polynomials belonging
to $\slist$, one obtains the following SOS decomposition:
$f = c_r s_r^2 + \dots + c_1 s_1^2$.
\begin{example}
\label{ex:2}
Let us consider the same polynomial
$f := \frac{1}{16} X^6 + X^4 - \frac{1}{9} X^3 - \frac{11}{10} X^2 +
\frac{2}{15} X + 2 \in \Q[X]$
as in Example~\ref{ex:1}.  We describe the different steps performed
by Algorithm~$\univsos1$:

\begin{itemize}
\item The polynomial $f$ is square-free so we obtain $p = f$
  (Step~\lineref{line:sqrfree3}). After performing the loop from
  Step~\lineref{line:epsi} to Step~\lineref{line:epsf},
  Algorithm~$\univsos2$ provides the value
  $\varepsilon = \frac{1}{32}$ at Step~\lineref{line:epsafe} as well
  as the polynomial
  $p_\varepsilon := p - \frac{1}{32} (1 + X^2 + X^4 + X^6)$ which has
  no real root.
\item Next, after increasing three times the precision in the loop
  going from Step~\lineref{line:deltai} to Step~\lineref{line:deltaf},
  the result of the approximate root computation yields
  $s_1 = X^3 - \frac{69}{8} X$ and
  $s_2 = 7 X^2 - \frac{1}{4} X - \frac{63}{8}$.
\end{itemize}

Applying Algorithm~$\univsos2$, we obtain the following two lists of
size $6 + 3 = 9 $:
\begin{align*}
\clist = & \biggl[\frac{1}{32},\frac{1}{32},\frac{913}{15360},\frac{731}{92160},\frac{7}{1152},\frac{1}{32},\frac{79}{7680},\frac{1}{576},0 \biggl] \,, \\ 
\slist = &
\biggl[X^3 - \frac{69}{8} X,
7 X^2 - \frac{1}{4} X - \frac{63}{8},  
1,
X,
X^2,
X^3,
X+ \frac{1}{2},
X (X - \frac{1}{2}),
X^2 (X + \frac{1}{2})
\biggr] 
\,, 
\end{align*}
yielding the following weighted SOS decomposition: 
\begin{align*}
f = 
\frac{1}{32} \biggl(X^3 - \frac{69}{8} X\biggr)^2 + \frac{1}{32} \biggl(7 X^2 - \frac{1}{4} X - \frac{63}{8}\biggr)^2 + \frac{913}{15360} + \frac{731}{92160} X^2 \\
+ \frac{7}{1152}X^4 + \frac{1}{32}X^6 + \frac{79}{7680} \biggl(X+ \frac{1}{2}\biggr)^2 + \frac{1}{576}X^2 \biggl(X - \frac{1}{2}\biggr)^2
 \,.
\end{align*}
\end{example}

\subsection{Bit size of the output}
\label{sec:bitsize2}
%
First, we need the following auxiliary result:
\begin{lemma}
\label{th:lowerbound}
Let $p \in \Z[X]$ be a positive polynomial over $\R$, with
$\deg p = n = 2 k$  
and $\tau$ be an upper bound on the bitsize of the
coefficients of $p$. 
Then, one has
\[
\inf_{x \in \R} p(x) 
> (n 2^{\tau})^{-n+2} 2^{-n \log_2 n - n \tau}
\,.
\]
\end{lemma}
\begin{proof}
  Denoting by $\tau'$ the maximum bit size of the coefficients of $p'$
  and instantiating $\alpha=\inf_{x\in \R} p(x)$ with a global
  minimizer of $p$, $Q$ with $p$ and $A$ with $p'$ in the third item
  of~\cite[Lemma 3.2]{Melczer16}, one obtains.
\[
\inf_{x \in \R} p(x) > (n 2^{\tau})^{-n+2} 2^{-n \tau'}
\]
Now, remark that $\tau' \leq \log_2 n + \tau$. Using this inequality
in the one above allows to conclude. 
\end{proof}

\begin{lemma}
\label{th:bitrooteps}
Let $p \in \Z[X]$ be a positive polynomial over $\R$, with
$\deg p = n = 2 k$ and let $\tau$ be an upper bound on the bitsize of
the coefficients of $p$. Then there exists a positive integer $N$ such
that for $\varepsilon := \frac{1}{2^N}$, the polynomial
$p_\varepsilon := p- \varepsilon \sum_{i=0}^k X^{2 i}$ is positive
over $\R$ and
$N = \tau (\varepsilon) \leq \bigo{(n \log_2 n + n
  \tau)}$.
\end{lemma}
\begin{proof}
  Let us first consider the polynomial
  $r := p - \frac{p_n}{2} \sum_{i=0}^k X^{2 i}$. Using~\cite[Corollary
  10.4]{BPR06}, the absolute value of each real root of the polynomial
  $r$ is bounded by $n \, 2^{\tau(r)} \leq 2 n 2^{\tau}$. By defining
  $R := 2 n 2^{\tau}$, it follows that the polynomial $r$ is positive
  for all $|x| > R$. In addition, for all positive integer $N$ and
  $\varepsilon = \frac{1}{2^N}$, one has {$\varepsilon \leq \frac{1}{2} \leq \frac{p_n}{2}$} and 
  $p_\varepsilon = p- \varepsilon \sum_{i=0}^k X^{2 i} \geq p -
  \frac{p_n}{2} \sum_{i=0}^k X^{2 i} = r$,
  which implies that the polynomial $p_\varepsilon$ is also positive
  for all $|x| > R$.
%
%
  Since {$R = 2 n 2^{\tau} > 1$}, one has
  $1 + R^2 \dots + R^n < n R^n$.  Let us choose the smallest positive
  integer $N$ such that $n R^n \leq 2^N \inf_{|x| \leq R} p$.  This
  implies that
  $\varepsilon < \frac{\inf_{|x| \leq R} p} {1 + R^2 \dots + R^n} $,
  ensuring that the polynomial $p_\varepsilon$ is also positive for
  all $|x| \leq R$. Applying Lemma~\ref{th:lowerbound}, we obtain the
  following upper bound: \if{
\[
2^N \leq n R^n \frac{2^{(2 n - 1) \tau(p)} (n+1)^{2 n - 1/2}}{3^{n/2} } 
= \frac{n^{n+1} 2^n 2^{n \tau(p)} 2^{(2 n - 1) \tau(p)} (n+1)^{2 n - 1/2}}{3^{n/2} } 
\,,
\]
}
\fi
\[
2^N \leq n R^n (n 2^{\tau})^{n-2} 2^{n \log_2 n + n \tau}  
= n^{n+1} 2^n 2^{n \tau} (n 2^{\tau})^{n-2} 2^{n \log_2 n + n \tau}. 
\]
The announced estimate follows straightforwardly. 
%
%
\end{proof}

In the sequel, we denote by $z_1,\dots, z_n$ the {(not necessarily distinct)} complex roots of the
polynomial $p_\varepsilon$. {Assuming that we approximate each complex root with a relative precision of $\delta$, 
  we write $\hat{z_1},\dots, \hat{z_n}$ for the approximate complex
  root values satisfying $ \hat{z_i} = z_i (1 + e_i)$, with
  $|e_i| \leq 2^{- \delta}$, for all $i=1,\dots,n$.}
\begin{theorem}
\label{th:bitsize2}
Let $f \in \Z[X]$ be a positive polynomial over $\R$, with
$\deg f = n$ and $\tau$ be an upper bound on the bitsize of the
coefficients of $f$.  Then the maximal bitsize of the coefficients
involved in the weighted SOS decomposition of $f$ obtained with
Algorithm~$\univsostwo$ is upper bounded by
$ \bigo{( n^3 + n^2 \tau )}$.
\end{theorem}
\begin{proof}
  Let $p$ be the square-free part of the polynomial $f$ (see
  Step~\lineref{line:sqrfree3} of Algorithm~$\univsostwo$). Then by
  using Lemma~\ref{th:bitdiv}, one has
  $\tau(p) \leq n + \tau + \log_2 (n+1) = \bigo{(n + \tau)}$.

  Let $\varepsilon = \frac{1}{2^N}$ as in Lemma~\ref{th:bitrooteps} so
  that the polynomial
  $p_\varepsilon = p- \varepsilon \sum_{i=0}^k X^{2 i}$ is positive
  over $\R$. By Lemma~\ref{th:bitrooteps}, one has
  $N = C(n^2 + n \tau)$ for some $C>1$.  Let us write
  $p_\varepsilon = \sum_{i=0}^n a_i X^i$ with $a_n = {\ell}$
  {and prove that a precision of
    $\delta := N + \log_2(5 n \|p\|_\infty) = C (n^2 + n \tau) +
    \log_2(5 n \|p\|_\infty)$
    is enough to ensure that the coefficients of $u$ satisfy
    $\varepsilon \geq \frac{|u_{2 i + 1}|}{4} - u_{2 i} + |u_{2 i -
      1}|$,
    for all $i=0,\dots, k$.  First, note that
    $e := 2^{-\delta} < \frac{1}{n (n+1)}$ holds.}  By using Vieta's
  formulas provided in Lemma~\ref{eq:vieta}, one has for all
  $j=1,\dots,n$:
\[
\sum_{1 \leq i_1 < \dots < i_j \leq n} z_{i_1} \dots z_{i_j} = (-1)^j \frac{a_{n-j}}{{\ell}} \,.
\]
Then one has for all $j=1,\dots,n$:
\[
 u_{n-j}  = {\ell}   \sum_{1 \leq i_1 < \dots < i_j \leq n} (z_{i_1} \dots z_{i_j} -  \hat{z}_{i_1} \dots \hat{z}_{i_j} ) =
  \sum_{1 \leq i_1 < \dots < i_j \leq n} z_{i_1} \dots z_{i_j}  \bigl(1 -  (1 + e_{i_1})\dots (1 + e_{i_j})  \bigr) \,.
\]
Since $e <  \frac{1}{n}$, one can apply~\cite[Lemma 3.3]{higham2002accuracy}, which yields $\prod_{1 \leq i_1 < \dots < i_j \leq n} (1 + e_{i_j}) \leq 1 + \theta_j$, with $|\theta_j| \leq \frac{j e}{1 - j e}$. In addition, one has $(j+1) e - \frac{j e}{1 - j e} = \frac{e (1 - j (j+1) e)}{1 - j e}  \geq 0$ since $e < \frac{1}{n (n+1)} < \frac{1}{j (j+1)}$, for all $j=1,\dots,n$. Hence, one has $|u_{n-j}| \leq |a_{n-j}| (j+1) e$, for all $j=1,\dots,n$. 

This implies that for all $i=0,\dots, k$:
\[
\frac{|u_{2 i + 1}|}{4} - u_{2 i} + |u_{2 i - 1}| \leq 
{e \| p_\varepsilon \|_{\infty} } \bigl(\frac{2n}{4} + 2n -1 + 2 n - 2 \bigr) \leq 5 n e \| p_\varepsilon \|_{\infty} \leq 5 n e \| p \|_{\infty} \,.
\]
%
{Since $\delta = N + \log_2(5 n \|p\|_\infty)$, one has $5 n e \| p \|_{\infty} = \varepsilon$. Thus, for all $i=0,\dots, k$, $\varepsilon \geq \frac{|u_{2 i + 1}|}{4} - u_{2 i} + |u_{2 i - 1}|$ holds with $\delta = \bigo{(n^2 + n \tau + \log_2 n + n + \tau)} = \bigo{(n^2 + n \tau)}$.}

For each $j=1, \dots, n$, choosing {$e_j = e = 2^{-\delta}$}
and $\hat{z}_j = z_j (1+2^{-\delta})$, yields
$|u_{n-j}| = | a_{n-j} | | 1 - (1 + 2^{-\delta})^j | $.  Next, we
bound the size of the weighted SOS decomposition.  One has
$\tau(\delta) = \bigo{ (n^2 + n \tau)}$ and for all $i=1,\dots, n$,
$\tau(a_{n-i}) \leq \tau (\varepsilon) = \bigo{( n^2 + n \tau)}$.
Therefore, for all $j=1, \dots, n$,
$\tau(u_{n-j}) \leq \bigo{(n^2 + n \tau + j (n^2 + n \tau) )}$ and the
maximal bitsize of the coefficients of $u$ is bounded by
$\bigo{(n^3 + n^2 \tau)}$.

From Lemma~\ref{th:cauchyint}, one has
$|\hat{z}_j | = | z_j | (1 + 2^{-\delta}) \geq
\frac{1}{2^{\tau(p_\varepsilon)} + 1} (1 + 2^{-\delta})|$,
so that it is enough to perform root isolation for the polynomial
$p_\varepsilon$ with a precision bounded by
$\bigo{( \tau(p_\varepsilon) + \delta)} = \bigo{(n^2 + n \tau)}$.


Finally, the weighted SOS decomposition of $f$ has coefficients of
maximal bitsize bounded by $\bigo{(n^3 + n^2 \tau)}$ as claimed.
\end{proof}

\subsection{Bit complexity analysis}
\label{sec:bitop2}
\begin{theorem}
\label{th:bitop2}
Let $f \in \Z[X]$ be a positive polynomial over $\R$, with
$\deg f = n = 2 k$ and $\tau$ be an upper bound on the bitsize of the
coefficients of $f$. Then, on input $f$, Algorithm~$\univsostwo$ runs in boolean time
$$ \bigotilde{(n^4 + n^3 \tau )}.$$ 
\end{theorem}
\begin{proof}
  By Lemma~\ref{th:costsqrfree}, the square-free decomposition of $f$
  can be computed using an expected number of $\bigotilde{(n^2\tau )}$
  boolean operations.
  {Checking that the polynomial $p_\varepsilon$ has no real
    root can be performed using an expected number of
    $\bigotilde{(n^2 \cdot \tau(\varepsilon))} = \bigotilde{(n^3
      \tau)}$
    boolean operations while relying on Sylvester-Habicht
    Sequences~\cite[Corollary~5.2]{Lickteig01}.}

  As seen in the proof of Theorem~\ref{th:bitsize2}, the complex roots
  of $p_\varepsilon$ must be approximated with isolating intervals
  (resp.~disks) of radius less than
  $2^{ - \tau(p_\varepsilon)- \delta}$.  Thus, by
  Lemma~\ref{th:rootisolation}, all real (resp.~complex) roots of
  $p_\varepsilon$ can be computed in
  $\bigotilde{(n^3 + n^2 \tau(p_\varepsilon) + n (\delta +
    \tau(p_\varepsilon) )} = \bigotilde{(n^4 + n^3 \tau )}$
  boolean operations.
  
  {As in the proof of Theorem~\ref{th:bitsize2}, one can select
    $|u_{n-j}| = | a_{n-j} | | 1 - (1 + 2^{-\delta})^j |$, for all
    $j=1, \dots, n$. This implies that the computation of each
    coefficient of $u$ can be performed with at most
    $\bigotilde{(n \cdot \tau(\delta))} =\bigotilde{(n^3 + n^2 \tau )}
    $
    boolean operations. Eventually, we obtain a bound of
    $\bigotilde{(n^4 + n^3 \tau )}$ for the computation of all
    coefficients of $u$, which yields the desired result.}
\end{proof}
We state now the complexity result for checking the SOS certificates
output by Algorithm $\univsostwo$. As for the output of
Algorithm~$\univsosone$, this is done through evaluation of the output
at $n+1$ distinct values where $n$ is the degree of the output. 
\begin{theorem}
\label{th:verif2}
Let $f \in \Z[X]$ be a positive polynomial over $\R$, with
$\deg f = n = 2 k$ and $\tau$ be an upper bound on the bitsize of the
coefficients of $f$.  Then one can check the correctness of the
weighted SOS decomposition of $f$ obtained with
Algorithm~$\univsostwo$ using $ \bigotilde{( n^4 + n^3 \tau )}$ bit
operations.
\end{theorem}
\begin{proof}
  From~\cite[Corollary 8.27]{Gathen99}, the cost of polynomial
  multiplication in $\Z[X]$ of degree less than $n $ with coefficients
  of bitsize upper bounded by $l$ is bounded by
  $\bigotilde{(n \cdot l)}$.  By Theorem~\ref{th:bitsize2}, the
  maximal coefficient bitsize of the SOS decomposition of $f$ obtained
  with Algorithm~$\univsostwo$ is upper bounded by
  $l = \bigo{( n^3 + n^2 \tau )}$.  Therefore, from~\cite[Corollary
  10.8]{Gathen99}, the cost of the evaluation of this decomposition at
  $n$ points can be performed using at most
  $\bigotilde{(n \cdot (n^3 + n \tau) )}$ boolean operations as
  claimed.
\end{proof}
\if{
\begin{remark}
  \label{rk:univsos2} \mohab{\bf bof bof pour cette remarque ; je
    sugg\`ere de tout faire avec des entr\'ees dans $\Z[X]$ et de dire
    dans l'intro que tout est valable avec des entr\'ees dans $\Q[X]$
    en se ramenant \`a $\Z[X]$}. A similar complexity analysis can be
  done while assuming that $f \in \Q[X]$.  Let us write the polynomial
  $f = \sum_{i=0}^{2 k} \frac{b_i}{c_i} X^i$, such that each rational
  coefficient $\frac{b_i}{c_i}$ is irreducible and let $l \in \Z$ be
  the least common multiples of all the $c_i$.  Then, one has
  $\tau(l) \leq \sum_{i=0}^{2 k} \tau(c_i) \leq (2k+1) \tau $,
  yielding
  $\tau (l f) \leq \tau(l) + \tau (f) \leq (2k+2) \tau = \bigo{(k
    \tau)}$.
  Therefore, we obtain an SOS decomposition whose total bitsize length
  is upper bounded by $\bigo{(n^4 \tau)}$ and this decomposition can
  be computed using an expected number of $\bigotilde{(n^4 \tau)}$ bit
  operations.
\end{remark}
}\fi

\section{Practical experiments}
\label{sec:benchs}
Now, we present experimental results obtained by applying
Algorithm~$\univsosone$ and Algorithm~$\univsostwo$, respectively
presented before in Sections~\ref{sec:method1}
and~\ref{sec:method2}. Both algorithms have been implemented in a
tool, called $\univsos$, written in Maple version 16. The interested
reader can find more details about installation and benchmark
execution on the dedicated
webpage.\footnote{\url{https://github.com/magronv/univsos}} This tool
is integrated to the RAGlib Maple
package\footnote{\url{http://www-polsys.lip6.fr/~safey/RAGLib/}}. We
obtained all results on an Intel Core i7-5600U CPU (2.60 GHz) with 16Gb of RAM. SOS
decomposition (resp.~verification) times are provided after averaging
over five (resp.~thousand) runs.

As mentioned in~\cite[Section 6]{Chevillard11}, the SOS decomposition
performed by Algorithm~$\univsostwo$ has been implemented using the
PARI/GP software tool\footnote{\url{http://pari.math.u-bordeaux.fr}}
and is freely
available.\footnote{\url{https://hal.archives-ouvertes.fr/ensl-00445343v2}}
To ensure fair comparison, we have rewritten this algorithm in
Maple. To compute approximate complex roots of univariate polynomials,
we rely on the PARI/GP procedure \texttt{polroots} through an
interface with our Maple library. We also tried to use the Maple
procedure \texttt{fsolve} but the \texttt{polroots} routine from
Pari/GPyields significantly better performance for the polynomials
involved in our examples.
%

The nine polynomial benchmarks presented in Table~\ref{table:bench1}
allow to approximate some given mathematical functions, considered
in~\cite[Section 6]{Chevillard11}. Computation and verification of SOS
certificates are a mandatory step required to validate the supremum
norm of the difference between such functions and their respective
approximation polynomials on given closed intervals. This boils down
to certify two inequalities of the form
$\forall x \in [b, c], p(x) \geq 0$, with $p \in \Q[X]$, $b, c \in \Q$
and $\deg p = n$. As recalled in~\cite[Section 5.2.5]{Chevillard11},
this latter problem can be addressed by computing a weighted SOS
decomposition of the polynomial
$q(Y) := (1+Y^2)^n \, p \Bigl( \frac{b+c Y^2}{1 + Y^2}\Bigr)$, with
either Algorithm~$\univsosone$ or Algorithm~$\univsostwo$.  For each
benchmark, we indicate in Table~\ref{table:bench1} the degree $n$ and
the bitsize $\tau$ of the input polynomial, the bitsize $\tau_1$ of
the weighted SOS decomposition provided by Algorithm~$\univsosone$ as
well as the corresponding computation (resp.~verification) time $t_1$
(resp.~$t_1'$). Similarly, we display $\tau_2, t_2, t_2'$ for
Algorithm~$\univsostwo$.
%
The table results show that for all other eight benchmarks,
Algorithm~$\univsostwo$ yields better certification and verification
performance, together with more concise SOS certificates. This
observation confirms what we could expect after comparing the
theoretical complexity results from Sections~\ref{sec:method1}
and~\ref{sec:method2}.

\begin{table*}[!ht]
\begin{center}
\caption{Comparison results of output size and performance between  Algorithm~$\univsosone$ and Algorithm~$\univsostwo$ for non-negative polynomial benchmarks from~\cite{Chevillard11}.}
\begin{tabular}{ccr|rrr|rrr}
\hline
\multirow{2}{*}{Id} & \multirow{2}{*}{$n$} & \multirow{2}{*}{$\tau$ (bits)} & \multicolumn{3}{c|}{$\univsosone$} & \multicolumn{3}{c}{$\univsostwo$} \\
 & & & $\tau_1$ (bits) & $t_1$ (ms) & $t_1'$ (ms) & $\tau_2$ (bits) & $t_2$ (ms) & $t_2'$ (ms) \\
\hline  
\# 1 & 13 & 22 682 & 3 403 218 & 2 723 & 0.40 & 51 992 & 824 & 0.14 \\
\# 3 & 32 & 269 958 & 11 613 480 & 13 109 &  1.18 & 580 335 & 2 640 & 0.68 \\
\# 4 & 22 & 47 019 &  1 009 507 & 4 063 & 1.45 & 106 797 & 1 776 & 0.31\\
\# 5 & 34 & 117 307 & 8 205 372 & 102 207  & 20.1 &   265 330 & 5 204 & 0.60 \\
\# 6 & 17 & 26 438 & 525 858 & 1 513 & 0.74 & 59 926 & 1 029 & 0.21 \\
\# 7 & 43 & 67 399 & 62 680 827 &  217 424 & 48.1 & 152 277 & 11 190 & 0.87 \\
\# 8 & 22 & 27 581 & 546 056 & 1 979 & 0.77 & 63 630 & 1 860 & 0.38 \\
\# 9 & 20 & 30 414 & 992 076 & 964 & 0.44 & 68 664 & 1 605 & 0.25 \\
\# 10 & 25 & 42 749 & 3 146 982 & 1 100 & 0.38 &  98 926 & 2 753 & 0.39 \\
\hline
\end{tabular}
\label{table:bench1}
\end{center}
\end{table*}

The comparison results available in Table~\ref{table:bench2} are obtained for power sums of increasing degrees. For a given natural integer $n = 2k$ with $10 \leq n \leq 500$, we consider the polynomial $P_n := 1 + X + \dots + X^n$. The roots of this polynomial are the $n+1$-th roots of unity, thus yielding the following  SOS decomposition with real coefficients: $P_n := \prod_{j=1}^k ((X - \cos \theta_j)^2 + \sin^2 \theta_j)$, with $\theta_j := \frac{2 j \pi}{n+1}$, for each $j=1, \dots, k$. 
By contrast with the benchmarks from Table~\ref{table:bench1}, Table~\ref{table:bench2} shows that Algorithm~$\univsosone$ produces output certificates of much smaller size compared to Algorithm~$\univsostwo$, with a bitsize ratio lying between $6$ and $38$ for values of $n$ between $10$ and $200$. This is due to the fact that Algorithm~$\univsosone$ outputs a value of t equal to 0 at each step.
The execution performance of Algorithm~$\univsosone$ are also much better in this case. The lack of efficiency of Algorithm~$\univsostwo$ is due to the computational bottleneck occurring in order to obtain accurate approximation of the relatively close roots $\cos \theta_j \pm i \sin \theta_j$, $j = 1, \dots, k$. For $n \geq 300$, execution of Algorithm~$\univsostwo$ did not succeed after two hours of computation, as meant by the symbol $-$ in the corresponding line.

\begin{table*}[!ht]
\begin{center}
\caption{Comparison results of output size and performance between  Algorithm~$\univsosone$ and Algorithm~$\univsostwo$ for non-negative power sums of increasing degrees.}
\begin{tabular}{c|rrr|rrr}
\hline
\multirow{2}{*}{$n$} & \multicolumn{3}{c|}{$\univsosone$} & \multicolumn{3}{c}{$\univsostwo$} \\
&  $\tau_1$ (bits) & $t_1$ (ms) & $t_1'$ (ms) & $\tau_2$ (bits) & $t_2$ (ms) & $t_2'$ (ms) \\
\hline  
10 & 84 & 7 & 0.03 & 
567 & 264 & 0.03\\
20 & 195 & 10 & 0.05 & 
1 598 & 485 & 0.06\\
40 & 467 & 26 & 0.09 & 
6 034 & 2 622 & 0.18\\
60 & 754 & 45 & 0.14 & 
12 326 & 6 320 & 0.32 \\
80 & 1 083 & 105 & 0.18 & 
21 230 & 12 153 & 0.47\\
100 & 1 411 & 109 & 0.26 & 
31 823 & 19 466 & 0.69 \\
200 & 3 211 & 444 & 0.48 & 
120 831 & 171 217 & 2.08 \\
300 & 5 149 & 1 218 & 0.74 & \multirow{4}{*}{$-$} & \multirow{4}{*}{$-$} & \multirow{4}{*}{$-$} \\
400 & 7 203 & 2 402 & 0.95 & & & \\
500 & 9 251 & 4 292 & 1.19 & & & \\
1000 & 20 483 & 30 738 & 2.56 & & & \\
\hline
\end{tabular}
\label{table:bench2}
\end{center}
\end{table*}

Further experiments are summarized in Table~\ref{table:bench3} for
modified Wilkinson polynomials $W_n$ of increasing degrees $n = 2k$
with $10 \leq n \leq 600$ and $W_n := 1 + \prod_{j=1}^k (X -j)^2$. The
complex roots $j \pm i$, $j=1,\dots,k$ of $W_n$ are relatively close,
which yields again a significant lack of performance of
Algorithm~$\univsostwo$. As observed in the case of power sums,
timeout behaviors occur for $n \geq 60$.  In addition, the bitsize of
the SOS decompositions returned by Algorithm~$\univsosone$ are much
smaller.  This is a consequence of the fact that in this case, $a = 1$
is the global minimizer of $W_n$. Hence, the algorithm always
terminates at the first iteration by returning the trivial quadratic
approximation $f_t = f_a = 1$ together with the square-free
decomposition of $W_n - f_t = \prod_{j=1}^k (X -j)^2$.

\begin{table*}[!ht]
\begin{center}
\caption{Comparison results of output size and performance between  Algorithm~$\univsosone$ and Algorithm~$\univsostwo$ for modified Wilkinson polynomials of increasing degrees.}
\begin{tabular}{cr|rrr|rrr}
\hline
\multirow{2}{*}{$n$} & \multirow{2}{*}{$\tau$ (bits)} & \multicolumn{3}{c|}{$\univsosone$} & \multicolumn{3}{c}{$\univsostwo$} \\
& &  $\tau_1$ (bits) & $t_1$ (ms) & $t_1'$ (ms) & $\tau_2$ (bits) & $t_2$ (ms) & $t_2'$ (ms) \\
\hline
10 & 140 & 47 & 17 & 0.01 & 2 373 & 751 & 0.03 \\
20 & 737 & 198 & 31 & 0.01 & 12 652 & 3 569  & 0.08 \\
40 & 3 692 & 939 & 35 & 0.01 & 65 404 & 47 022  &  0.17\\
60 & 9 313 & 2 344 & 101 & 0.01 & \multirow{8}{*}{$-$} & \multirow{8}{*}{$-$} & \multirow{8}{*}{$-$} \\
80 & 17 833 & 4 480 & 216 & 0.01 &  & & \\
100 & 29 443 & 7 384 & 441 & 0.01 & & &  \\
200 & 137 420 & 34 389 & 3 249 & 0.01 & & & \\
300 & 335 245 & 83 859 & 11 440 & 0.01 & & &  \\
400 & 628 968 & 157 303 & 34 707 & 0.02 & & &  \\
500 & 1 022 771 & 255 767 & 73 522 & 0.02 & & & \\
600 & 1 519 908 & 380 065 & 149 700 & 0.04 & & & \\
\hline  
\end{tabular}
\label{table:bench3}
\end{center}
\end{table*}
Finally, we consider experimentation performed on modified Mignotte polynomials defined by $M_{n,m} := X^n + 2 (101 X - 1)^{m}$ and $N_n := (X^n + 2 (101 X - 1)^2) (X^n + 2 ((101 + \frac{1}{101})X - 1)^2)$, for even natural integers $n$ and $m \leq 2$.
The corresponding results are displayed in Table~\ref{table:bench4} for $M_{n,m}$ with $m = 2$ and $10 \leq n \leq 10 000$, $m = n -2$ and $10 \leq n \leq 100$ as well as for $N_n$ with $10 \leq n \leq 100$. Note that similar benchmarks are used in~\cite{RootIsol11} to anayze the efficiency of (real) root isolation techniques over polynomial with relatively close roots.
As for modified Wilkinson polynomials, Algorithm~$\univsostwo$ can only handle small size instances, due to  limited scalablity of the \texttt{polroots} procedure. 
In this case $a = \frac{1}{100}$ is the unique global minimizer of $M_{n,2}$. Thus, Algorithm~$\univsosone$ always outputs weighted SOS decompositions of polynomials $M_{n,2}$ within a single iteration by first computing the quadratic polynomial $f_t = f_a = 2 (101 X - 1)^2$ and the trivial square-free decomposition $W_n - f_t = X^n$. In the absence of such minimizers, Algorithm~$\univsosone$ can only handle instances of polynomials $M_{n,n-2}$ and $N_n$ with moderate degree (less than 100).


\begin{table*}[!ht]
\begin{center}
\caption{Comparison results of output size and performance between  Algorithm~$\univsosone$ and Algorithm~$\univsostwo$ for modified Mignotte polynomials of increasing degrees.}
\begin{tabular}{ccr|rrr|rrr}
\hline
\multirow{2}{*}{Id} & \multirow{2}{*}{$n$} & \multirow{2}{*}{$\tau$ (bits)} & \multicolumn{3}{c|}{$\univsosone$} & \multicolumn{3}{c}{$\univsostwo$} \\
& & &  $\tau_1$ (bits) & $t_1$ (ms) & $t_1'$ (ms) & $\tau_2$ (bits) & $t_2$ (ms) & $t_2'$ (ms) \\
\hline
\multirow{4}{*}{$M_{n,2}$} & 10~ & \multirow{4}{*}{27} & \multirow{4}{*}{23} & 2 & \multirow{4}{*}{0.01} &  4 958 & 1 659 & 0.04 \\
& $10^2$ &  &  & 3 &   & \multirow{3}{*}{$-$} & \multirow{3}{*}{$-$} & \multirow{3}{*}{$-$} \\
& $10^3$ &  &  & 85 &  & &  &  \\
& $10^4$ &  &  & 3 041 &  & & &  \\  
\hline  
\multirow{5}{*}{$M_{n,n-2}$} & 10 & 288 & 25 010 & 21 & 0.03 & 
6 079 & 2 347 & 0.04 \\
& 20 & 1 364 & 182 544 & 138 & 0.04 & 
26 186 & 10 922 & 0.06 \\
& 40 & 5 936 & 1 365 585 & 1 189 & 0.13 & \multirow{3}{*}{$-$} & \multirow{3}{*}{$-$} & \multirow{3}{*}{$-$} \\
& 60 & 13 746 & 4 502 551 & 4 966 & 0.33 & &  & \\
& 100 & 39 065 & 20 384 472  & 38 716 & 1.66 & & &  \\
\hline  
\multirow{5}{*}{$N_n$} & 10 & \multirow{5}{*}{212} &   25 567 & 27 & 0.04 & \multirow{5}{*}{$-$} & \multirow{5}{*}{$-$}   & \multirow{5}{*}{$-$}   \\
& 20 & & 189 336 & 87 & 0.05 & & &  \\
& 40 & & 5 027 377 & 1 704 & 0.17 &  & & \\
& 60 & & 16 551 235  & 8 075 & 0.84 &  &  &  \\
& 100 & & 147 717 572 & 155 458  & 11.1 & & & \\
\hline  
\end{tabular}
\label{table:bench4}
\end{center}
\end{table*}

\section{Conclusion and perspectives}
We presented and analyzed two different algorithms $\univsosone$ and $\univsostwo$ to compute weighted sums of squares (SOS) decompositions of non-negative univariate polynomials. When the input polynomial has rational coefficients, one feature shared by both algorithms is their ability to provide non-negativity certificates whose coefficients are also rational.
Our study shows that the complexity analysis of Algorithm~$\univsosone$ yields an upper bound which is exponential w.r.t.~the input degree, while the complexity of Algorithm~$\univsostwo$ is polynomial.
However, comparison benchmarks emphasize the need for both algorithms to handle various classes of non-negative polynomials, e.g.~in the presence of rational global minimizers or when root isolation can be performed efficiently. 

A first direction of further research is a variant of 
Algorithm~$\univsostwo$ where one would compute approximate SOS decompositions of perturbed positive polynomials by using semidefinite programming (SDP) instead of root isolation. Preliminary 
experiments yield very promising results when the bitsize of the polynomials is small, e.g.~for power sums of degree up to 1000. However, the performance decrease when the bitsize becomes larger, either for polynomial benchmarks from~\cite{Chevillard11} or modified Wilkinson polynomials. At the moment, we are not able to provide any SOS decomposition for all such benchmarks.
Our SDP-based algorithm relies on the high-precision solver SDPA-GMP~\cite{Nakata10GMP} but it is still challenging to obtain precise values of eigenvalues/vectors of SDP output matrices. Another advantage of this technique is its ability to perform global polynomial optimization. A topic of interest would be to obtain the same feature with the two current algorithms. We also plan to develop extensions to the non-polynomial case.

\bibliographystyle{plain}

\end{document}